\documentclass{amsart}
\usepackage[foot]{amsaddr}

  \usepackage{cite}
  \usepackage[T1]{fontenc}        

\usepackage{amsmath}
\usepackage{amssymb}
\usepackage{amsfonts}
\usepackage{amsthm}

\usepackage{comment}

\usepackage[boxed]{algorithm}
\usepackage[noend]{algorithmic}
  \newlength\commentspace
  \newcommand\algcomment[2]{\makebox[0pt][l]{\hspace{-#1em}%
    \hspace{\commentspace}$\triangleright$ #2}}
\newcommand{\algorithmicfunc}[1]{\textsc{#1}}
\newcommand{\FUNC}[1]{\item[\algorithmicfunc{#1}]}

\floatname{algorithm}{Algorithm}

\theoremstyle{plain}
\newtheorem{theorem}{Theorem}[section]
\newtheorem{corollary}[theorem]{Corollary}
\newtheorem{proposition}[theorem]{Proposition}
\newtheorem{lemma}[theorem]{Lemma}
\newtheorem{question}{Question}

\theoremstyle{remark}
\newtheorem*{remark}{Remark}
\newtheorem*{remarks}{Remarks}

\DeclareSymbolFont{rsfscript}{OMS}{rsfs}{m}{n}
\DeclareSymbolFontAlphabet{\mathrsfs}{rsfscript}

\DeclareMathOperator{\rt}{\mathrm{rt}}

\usepackage{pstricks,pst-node,pst-text,pst-3d}
\usepackage{color}

\usepackage{tikz}
\usetikzlibrary{arrows.meta,automata}
\usetikzlibrary{shapes.geometric}
\usetikzlibrary{fit}
\usetikzlibrary{positioning,decorations.pathreplacing,calligraphy}

\newcommand{\mA}{\mathrsfs{A}}
\newcommand{\mB}{\mathrsfs{B}}
\newcommand{\mC}{\mathrsfs{C}}
\newcommand{\mD}{\mathrsfs{D}}
\newcommand{\mE}{\mathrsfs{E}}
\newcommand{\mF}{\mathrsfs{F}}

\newcommand{\mL}{\mathrsfs{L}}

\newcommand{\gR}{\mathrel{\mathfrak{R}}}
\newcommand{\gL}{\mathrel{\mathfrak{L}}}
\newcommand{\gD}{\mathrel{\mathfrak{D}}}

\newcommand{\sa}{synchronizing automata}
\newcommand{\san}{synchronizing automaton}

\newcommand{\sDFA}{synchronizing DFA}
\newcommand{\rl}{reset threshold}
\newcommand{\sw}{reset word}

\newcommand{\scc}{strongly connected component}

\makeatletter

\@namedef{subjclassname@2020}{\textup{2020} Mathematics Subject Classification}

\renewcommand{\subjclassname}{\textup{2020} Mathematics Subject Classification}

\renewcommand*\subjclass[2][2020]{%
  \def\@subjclass{#2}%
  \@ifundefined{subjclassname@#1}{%
    \ClassWarning{\@classname}{Unknown edition (#1) of Mathematics Subject Classification; using '2020'.}%
  }{%
    \@xp\let\@xp\subjclassname\csname subjclassname@#1\endcsname
  }%
}

\makeatother

\usepackage{mathptmx}      
\allowdisplaybreaks

\usepackage{pgfplots}
\pgfplotsset{compat=1.16}

\title{Adversarial Synchronization}

\author{Anton E. Lipin}
\email{tony.lipin@yandex.ru}
\author{Mikhail V. Volkov}
\email{m.v.volkov@urfu.ru}
\address[A. E. Lipin]{\normalfont Krasovskii Institute of Mathematics and Mechanics, 620108 Yekaterinburg, Russia; Ural Federal University, 620062 Yekaterinburg, Russia}
\address[M. V. Volkov]{\normalfont 620075 Yekaterinburg, Russia}

\date{Version of January 22, 2026}

\begin{document}

\begin{abstract}
We study a variant of the synchronization game on finite deterministic automata. In this game, Alice chooses one input letter of an automaton $\mA$ on each of her moves while Bob may respond with an arbitrary finite word over the input alphabet of~$\mA$; Alice wins if the word obtained by interleaving her letters with Bob’s responses resets $\mA$. We prove that if Alice has a winning strategy in this game on $\mA$, then $\mA$ admits a reset word whose length is strictly smaller than the number of states of $\mA$. In contrast, for any $k\ge 1$, we exhibit automata with shortest reset-word length quadratic in the number of states, on which Alice nevertheless wins a version of the game in which Bob's responses are restricted to arbitrary words of length at most $k$. We provide polynomial-time algorithms for deciding the winner in various synchronization games, and we analyze the relationships between variants of synchronization games on fixed-size automata.
\end{abstract}



\maketitle

\section{Introduction}
\label{intro}

\subsection{Automata} All automata in this paper are complete, deterministic, and finite. Such an automaton (a \emph{DFA}) is a pair $(Q,\Sigma)$ of two finite non-empty sets equipped with a map $Q\times\Sigma\to Q$. We call $Q$ the \emph{state set} and $\Sigma$ the \emph{input alphabet}; elements of $Q$ and $\Sigma$ are referred to as \emph{states} and \emph{letters}, respectively. The map $Q\times\Sigma\to Q$ is called the \emph{action} of letters on states. For a state $q\in Q$ and a letter $a\in\Sigma$, the result of the action of $a$ on $q$ is often denoted by just $qa$ in the literature. In this paper, we denote it by $q{\cdot}a$ by default, although we may also use  $\circ$ to denote the action, when working with multiple automata in one argument.

Any DFA $\mA=(Q,\Sigma)$ can be represented as a labeled digraph with vertex set $Q$ and edges of the form $q\xrightarrow{a}q{\cdot}a$ for all $q\in Q$ and $a\in\Sigma$. We will take the liberty of applying graph-theoretical notions to automata; that is, we speak, for instance, of paths or \scc{}s of $\mA$, meaning those of the digraph of $\mA$.

From time to time throughout this paper, we will encounter automata from the series $\{\mathrsfs{C}_{n}\}_{n=2,3,\dots}$ invented by Jan \v{C}ern\'{y}~\cite{Cerny:1964}, so we introduce them right now to exemplify the above concepts. The automaton $\mathrsfs{C}_n$ has the set $\mathbb{Z}_n$ of the residues modulo $n$ as the state set and two input letters, $a$ and $b$, that act as follows:
\[
0{\cdot} a:=1,\ \ m{\cdot} a:=m\ \text{ for \ $0<m<n$,}\qquad m{\cdot} b:=m+1\kern-8pt\pmod{n}.
\]
(Here and below, the symbol $:=$ signifies equality by definition; that is, $A:=B$ means that $A$ is defined as $B$.)

Figure~\ref{fig:cerny-n} shows a generic automaton from the \v{C}ern\'{y} series\footnote{We adopt the usual convention that edges with multiple labels represent bunches of parallel edges. In particular, the edge $0\xrightarrow{a,b}1$ in Figure~\ref{fig:cerny-n} represents the parallel edges $0\xrightarrow{a}1$ and $0\xrightarrow{b}1$.}.
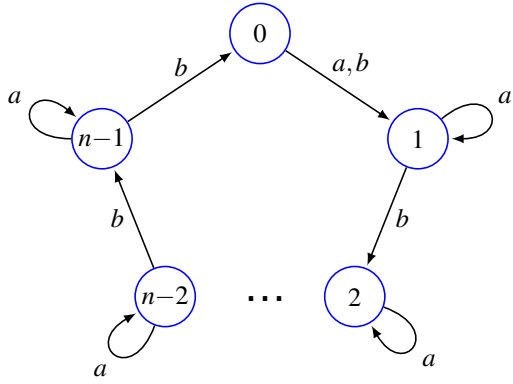
\begin{figure}[ht]
\begin{center}
\begin{tikzpicture}[scale=0.7,->,>=latex,shorten >=1pt,auto,semithick]
  \tikzstyle{every state}=[draw=blue,minimum size=8mm,inner sep=0pt]

  \node[state] (n0) at (0,-1) {0};
  \node[state] (n1) at (-3,-3) {$n{-}1$};
  \node[state] (n2) at (3,-3) {1};
  \node[state] (n3) at (-1.8,-6) {$n{-}2$};
  \node[state] (n4) at (1.8,-6) {2};

  \path (n1) edge node[above] {$b$} (n0)
        (n0) edge node[above, xshift=1ex] {$a,b$} (n2)
        (n2) edge node[right] {$b$} (n4)
        (n3) edge node[left] {$b$} (n1);

\path
 (n2) edge[-latex, loop, out = 40, in = 0, looseness=8] node[above right] {$a$} (n2)
 (n3) edge[-latex, loop, out = -110, in = -150, looseness=8] node[below left] {$a$} (n3)
 (n4) edge[-latex, loop, out = -20, in = -60, looseness=8] node[below right] {$a$} (n4)
 (n1) edge[-latex, loop, out=180, in=140, looseness=8] node[above left] {$a$} (n1);

  \node at (.1,-6.2) {\raisebox{10pt}{$\mathbf{\centerdot\,\centerdot\,\centerdot}$}};
\end{tikzpicture}
\end{center}
\caption{The automaton $\mathrsfs{C}_n$}\label{fig:cerny-n}
\end{figure}

A \emph{word} over $\Sigma$  is a finite sequence of letters. If $w=a_1\cdots a_\ell$, where $a_1,\dots a_\ell$ are letters, the number $\ell$ is called the \emph{length} of the word $w$ and is denoted by $|w|$. We allow the \emph{empty word}, denoted by $\varepsilon$, which is the empty sequence, and set $|\varepsilon|:=0$.

We use the standard notation $\Sigma^*$ for the set of all words over $\Sigma$. The action of letters on the states of a DFA $(Q,\Sigma)$ naturally extends to actions of words over $\Sigma$ on states and on sets of states. Namely, for $q\in Q$ and $w\in\Sigma^*$, the action of $w$ on $q$ is defined recursively as
\[
q{\cdot}w:=\begin{cases}
               q & \text{if } w=\varepsilon,\\
               (q{\cdot}w'){\cdot}a & \text{if } w=w'a \text{ for some } w'\in\Sigma^* \text{ and } a\in\Sigma,
             \end{cases}
\]
and the action of $w$ on a subset $P\subseteq Q$ is defined by $P{\cdot}w:=\{p{\cdot}w : p\in P\}$.

\subsection{Synchronizing automata} A DFA $(Q,\Sigma)$ is called \emph{synchronizing} if there exists a word $w\in\Sigma^*$ such that applying $w$ to any state in $Q$ results in the same state:
\[
q{\cdot}w = q'{\cdot}w \ \text{ for all } \ q,q'\in Q.
\]
Any word $w$ satisfying this condition is called a \emph{reset} word for the automaton. The minimum length of reset words for a \san{} $\mA$ is called its \emph{reset threshold} and is denoted by $\rt(\mA)$.

The DFAs $\mC_n$ from the \v{C}ern\'{y} series are well-known examples of \sa. \v{C}ern\'{y} \cite[Lemma 1]{Cerny:1964} proved that the shortest reset word for the automaton $\mathrsfs{C}_{n}$ is the word $(ab^{n-1})^{n-2}a$ of length $(n-1)^2$. Hence, $\rt(\mC_n)=(n-1)^2$.

Synchronizing automata are widely studied as transparent and natural models of error-resistant systems, with numerous applications in different areas of computer science, including coding theory, robotics, and the testing of reactive systems. They also reveal fascinating connections with symbolic dynamics, substitution systems, permutation groups, and other areas of mathematics. Per se, \sa{} attract considerable attention because of the \v{C}ern\'y conjecture, arguably the most longstanding open problem in the combinatorial theory of finite automata. (The conjecture, formulated in the late 1960s, states that the \rl{} of every \san{} with $n$ states does not exceed $\rt(\mC_n)=(n-1)^2$.) We refer the reader to chapter~\cite{KV} of the `Handbook of Automata Theory' and to the overview article \cite{Vo22} for an introduction to the field of automata synchronization. In connection with the \v{C}ern\'y conjecture, the `living survey' \cite{List} provides a regularly updated reference source.

\subsection{Synchronization games} F\"odor Fominykh and the second-named author \cite{FV} initiated the study of \sa{} from a game-theoretic perspective, motivated, on the one hand, by a game-theoretical approach to software testing proposed in~\cite{BGNV06} and, on the other hand, by a peg-solitaire-like game used in~\cite{AVZ06,AVZ07} to justify new examples of \sa{} whose \rl{}s are close to the \v{C}ern\'y conjecture bound. In the game introduced in~\cite{FV}, two players, Alice (Synchronizer) and Bob (Desynchronizer), take turns choosing letters from the input alphabet of a~DFA $\mathrsfs{A}$. Alice who wants to synchronize $\mathrsfs{A}$ wins when the word formed by the letters chosen by her and Bob, taken in the order in which they were selected, resets $\mA$. Bob, on the other hand, strives to prevent synchronization or, if synchronization is unavoidable, to delay it for as long as possible. Assuming both players play optimally, the outcome of the game depends solely on the automaton. This raises the problem of classifying \sa{} into those for which Alice has a winning strategy and those for which Bob does. Several results in this direction were obtained in~\cite{FV,FMV,MonNol18,FUN22,NCMA,JALC}.

\subsection{New synchronization games and our contribution} It was observed in \cite[Section 5.4]{JALC} that many of the results established in~\cite{FV,FMV,FUN22,NCMA,JALC} remain valid for a modified game in which Bob is allowed to choose an arbitrary word in response to each of Alice's moves. Synchronization games under such modified rules appear to better model reliably controlled systems that must perform tasks in the presence of natural or artificial interference. Indeed, nature or an adversary does not necessarily act in the same rhythm as we do, nor are they obliged to respond to each of our actions with exactly one obstacle or counteraction.

In this paper, we study synchronization games with various `degrees of freedom' for an adversary. Our main interest is in the speed of synchronization. Automata on which Alice can win playing against an adversary appear to be more amenable to synchronization. Consequently, it is natural to expect that, in the absence of resistance, they should be quickly synchronizable---that is, they should possess short reset words. We show that the \rl{} of a DFA $(Q,\Sigma)$ on which Alice can win even if Bob may respond with any word in $\Sigma^*$ is always less than $|Q|$. As a corollary, we obtain a new, fully combinatorial proof (and a strengthening) of a result by Jorge Almeida and Benjamin Steinberg \cite{AlSt09} on the \rl{} of \sa{} subject to certain algebraic restrictions. In contrast, for any $k$, we exhibit automata with a \rl{} quadratic in the number of states on which Alice wins if Bob may respond with any word of length less than $k$. For $k=2$, this resolves a question left open in \cite{NCMA}.

We also provide polynomial-time algorithms that decide which player, Alice or Bob, has a winning strategy on a given DFA in the variants of synchronization games we consider, and we analyze the relationships between these variants when played on fixed-size automata.

\subsection{Prerequisites} A fair effort has been made to keep the paper reasonably self-con\-tained. We use only a few standard notions from graph theory and algorithmics; no game- or automata-theoretic background is presupposed.

\section{Variants of adversarial synchronization}
\label{sec:rules}

\subsection{Definition and graphical representation} $\mathbb{N}$ is the set of all positive integers, and $\overline{\mathbb{N}}$ denotes $\mathbb{N}$ augmented with the symbol $\omega$. We extend the usual order on  $\mathbb{N}$ to $\overline{\mathbb{N}}$ by declaring $k<\omega$ for all $k\in\mathbb{N}$. For $k\in\overline{\mathbb{N}}$ and an alphabet $\Sigma$, we write $\Sigma^{<k}$ for the set of all words over $\Sigma$ of length less than $k$. In particular, $\Sigma^{<1}=\{\varepsilon\}$ and $\Sigma^{<\omega}=\Sigma^*$.

For $k\in\overline{\mathbb{N}}$ and a DFA $\mA = (Q,\Sigma)$, we call the following game the $k$-\emph{synchronization game} on $\mA$. Two players, Alice (Synchronizer) and Bob (Desynchronizer), take turns moving. On her turns, Alice chooses letters from $\Sigma$. Bob, on his turns, chooses words from $\Sigma^{<k}$ (thus, when $k=\omega$, Bob may choose arbitrary finite words).

The concatenation of all letters chosen by Alice and all words chosen by Bob up to some turn $t$, taken in the order in which they were selected, is called the \emph{history} of the game up to turn $t$. If, at some turn, this history is a reset word for the automaton $\mA$, then Alice wins the game. The goal of Bob is to continue the game indefinitely, or if synchronization is unavoidable, to prolong the game for as long as possible.

The reader may have noticed that we have not yet specified who starts the game, Alice or Bob. This does not actually matter. If, given a DFA, Alice wins a $k$-synchronization game in which she makes the first move, then she can also win the game in which Bob moves first, because a reset word remains a reset word when multiplied on the left by an arbitrary word. If, given a DFA, Bob has a winning strategy in a $k$-synchronization game which Alice starts, then he can also win the game in which he moves first by passing the first turn to Alice, that is, by choosing the empty word as his initial move.

For brevity, the following two conventions will be used throughout this paper. First, we will tacitly assume that DFAs under consideration have at least three states, since Alice trivially wins on each \san{} with two or less states. Second, since no games other than synchronization games are considered, we will shorten the expression `$k$-synchronization game' to `$k$-game'.

Using the standard graphical representation, one can visualize the $k$-game on a DFA $\mA$ as a two-player peg-solitaire-like game with (the digraph of) $\mA$ as the game board. At the beginning each state of $\mA$ holds one token. Tokens are moved and removed as follows. If Alice chooses a letter $a$, all tokens still present on the board simultaneously slide along the edges labeled $a$. Similarly, if Bob chooses a word $w$, all tokens still present on the board simultaneously slide along the paths formed by consecutive edges whose labels form $w$. Whenever several tokens arrive at the same state after this, all of them except one are removed, so that after the move is completed, each state holds at most one token.

We refer to the set of states carrying tokens after some turn $t$ as the \emph{position} at turn $t$. Since token movements correspond to actions of letters and words, the position at turn~$t$ in the $k$-game on a DFA $(Q,\Sigma)$ equals $Q{\cdot}u$, where $u$ is the history of the game up to turn $t$.

\begin{figure}[bht]
\begin{center}
\begin{tikzpicture}[scale=0.064,>=latex,auto,
  every node/.style={circle,minimum size=8mm,inner sep=0pt},
  token/.style={circle,minimum size=5mm,inner sep=0pt},
  b/.style={minimum size=4mm,inner sep=0pt}]

\node[b] (b0) at (37,30) {};
\node[b] (b1) at (28,1) {};
\draw[->,thick] (b0) -- node[sloped, above, yshift=-1.5ex]{Move $b$} (b1);

\node[b] (c0) at (101,30) {};
\node[b] (c1) at (111,1) {};
\draw[->,thick] (c0) -- node[sloped, above, yshift=-2ex]{Move $ab$} (c1);


\node[token,fill=gray!60] (t10) at (66,54) {};
\node[draw=blue] (n10) at (66,54) {0};
\node[draw=blue] (n11) at (42,38) {4};
\node[token, fill=gray!60] (t12) at (90,38) {};
\node[draw=blue] (n12) at (90,38) {1};
\node[token, fill=gray!60] (t13) at (54,18) {};
\node[draw=blue] (n13) at (54,18) {3};
\node[draw=blue] (n14) at (78,18) {2};

\draw[->] (n10) -- node[sloped, above, yshift=-1ex] {$a,b$} (n12);
\draw[->] (n12) -- node[sloped, below, yshift=1ex] {$b$} (n14);
\draw[->] (n14) -- node[sloped, below, yshift=1ex] {$b$} (n13);
\draw[->] (n13) -- node[sloped, below, yshift=1ex] {$b$} (n11);
\draw[->] (n11) -- node[sloped, above, yshift=-1ex] {$b$} (n10);

\draw[->] (n11) to [out=170,in=130,looseness=8] node[above,yshift=-1mm] {$a$} (n11);
\draw[->] (n12) to [out=50,in=10,looseness=8] node[above,yshift=-1mm] {$a$} (n12);
\draw[->] (n13) to [out=-130,in=-170,looseness=8] node[below,yshift=1mm] {$a$} (n13);
\draw[->] (n14) to [out=-10,in=-50,looseness=8] node[below,yshift=1mm] {$a$} (n14);


\node[draw=blue] (n0) at (27,-4)  {0};
\node[token, fill=gray!60] (t1) at (3,-20) {};
\node[draw=blue] (n1) at (3,-20) {4};
\node[token, fill=gray!60] (t2) at (51,-20) {};
\node[draw=blue] (n2) at (51,-20) {1};
\node[draw=blue] (n3) at (15,-40) {3};
\node[token, fill=gray!60] (t4) at (39,-40) {};
\node[draw=blue] (n4) at (39,-40) {2};

\draw[->] (n0) -- node[sloped, above, yshift=-1ex] {$a,b$} (n2);
\draw[->] (n2) -- node[sloped, below, yshift=1ex] {$b$} (n4);
\draw[->] (n4) -- node[sloped, below, yshift=1ex] {$b$} (n3);
\draw[->] (n3) -- node[sloped, below, yshift=1ex] {$b$} (n1);
\draw[->] (n1) -- node[sloped, above, yshift=-1ex]  {$b$} (n0);

\draw[->] (n1) to [out=170,in=130,looseness=8] node[above,yshift=-1mm] {$a$} (n1);
\draw[->] (n2) to [out=50,in=10,looseness=8] node[above,yshift=-1mm] {$a$} (n2);
\draw[->] (n3) to [out=-130,in=-170,looseness=8] node[below,yshift=1mm] {$a$} (n3);
\draw[->] (n4) to [out=-10,in=-50,looseness=8] node[below,yshift=1mm] {$a$} (n4);


\node[draw=blue] (n5) at (112,-4) {0};
\node[token, fill=gray!60] (t6) at (88,-20) {};
\node[draw=blue] (n6) at (88,-20) {4};
\node[draw=blue] (n7) at (138,-20) {1};
\node[draw=blue] (n8) at (100,-40) {3};
\node[token, fill=gray!60] (t9) at (124,-40) {};
\node[draw=blue] (n9) at (124,-40) {2};

\draw[->] (n5) -- node[sloped, above, yshift=-1ex] {$a,b$} (n7);
\draw[->] (n7) -- node[sloped, below, yshift=1ex] {$b$} (n9);
\draw[->] (n9) -- node[sloped, below, yshift=1ex] {$b$} (n8);
\draw[->] (n8) -- node[sloped, below, yshift=1ex] {$b$} (n6);
\draw[->] (n6) -- node[sloped, above, yshift=-1ex]  {$b$} (n5);

\draw[->] (n6) to [out=170,in=130,looseness=8] node[above,yshift=-1mm] {$a$} (n6);
\draw[->] (n7) to [out=50,in=10,looseness=8] node[above,yshift=-1mm] {$a$} (n7);
\draw[->] (n8) to [out=-130,in=-170,looseness=8] node[below,yshift=1mm] {$a$} (n8);
\draw[->] (n9)  to [out=-10,in=-50,looseness=8] node[below,yshift=1mm] {$a$} (n9);

\end{tikzpicture}
\caption{Moves in a synchronization game on $\mC_5$} \label{fig:moves}
\end{center}
\end{figure}
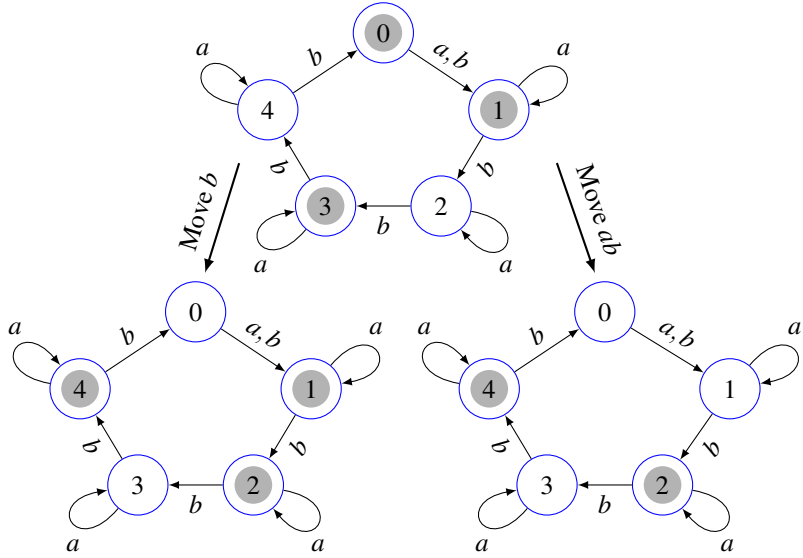

For illustration, the upper part of Figure~\ref{fig:moves} shows a typical intermediate position in a game on the automaton $\mC_5$ from the \v{C}ern\'{y} series. In this position, states 0, 1, and 3 hold tokens (shown in gray). The lower left part shows the effect of the move $b$ in this position while the lower right part demonstrates the result of the move $ab$. Observe that in the latter case one token has been removed, because the tokens originating from states 0 and 1 have moved to the same state.

In the graphical representation, Alice wins the game when all but one token have been removed. Bob wins if he can keep at least two tokens on the board indefinitely.

\subsection{Hierarchy of $A_k$-automata} Provided that both Alice and Bob play optimally, the outcome of the $k$-game on a DFA $\mA$ depends only on $\mA$. If Alice has a winning strategy in the $k$-game on $\mA$, we say that $\mA$ is an $A_k$-\emph{automaton}.
	
The following observations are immediate:
\begin{enumerate}
	\item[\rm(O1)] a DFA is an $A_1$-automaton if and only if it is synchronizing;
	
	\item[\rm(O2)] for each $k\in\mathbb{N}$, every $A_{k+1}$-automaton is an $A_k$-automaton;
	
	\item[\rm(O3)] a DFA is an $A_\omega$-automaton  if and only if it is an $A_k$-automaton for each $k\in\mathbb{N}$.
\end{enumerate}

For $k\in\overline{\mathbb{N}}$, denote the class of all $A_k$-automata by $\mathbf{A}_k$. By (O2) and (O3), the classes $\mathbf{A}_k$ form a decreasing hierarchy:
\begin{equation}\label{eq:hierarchy}
\mathbf{A}_1\supseteq\mathbf{A}_2\supseteq\dots\supseteq\mathbf{A}_k\supseteq\dots \mathbf{A}_\omega.
\end{equation}
In fact, all inclusions in \eqref{eq:hierarchy} are strict. To show this, consider, for each $n\ge 3$, the DFA $\mE_n=(\mathbb{Z}_n,\{b,c,d\})$ whose letters act as follows:
\begin{equation}\label{eq:rulesEn}
m{\cdot} b:=m+1\kern-7pt\pmod{n}; \qquad  m{\cdot} c:=\begin{cases}
			0 &\text{if } m = 0, \\
			1 &\text{if } m\ne 0;
		\end{cases}
\qquad  m{\cdot} d:=\begin{cases}
			1 &\text{if } m = n-1, \\
			0 &\text{if } m\ne n-1.
		\end{cases}
\end{equation}
Figire~\ref{fig:strict} shows the automata $\mE_3$ and $\mE_4$.

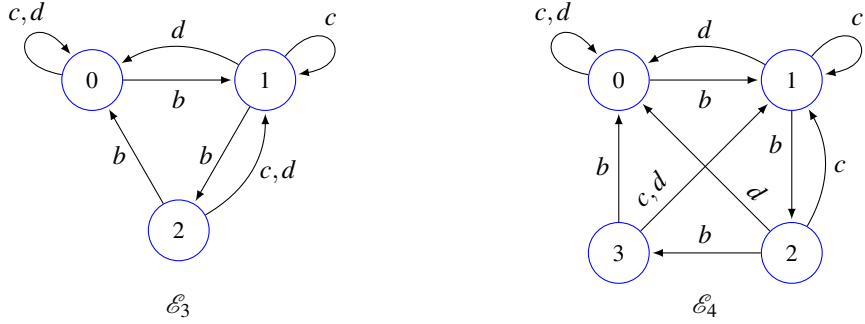
\begin{figure}[hbt]
\begin{center}
\begin{tikzpicture}[->,>=latex,shorten >=1pt,auto,
  every state/.style={circle,draw=blue,minimum size=23pt,font=\small}]

\node[state] (0) at (0,0) {0};
\node[state] (1) at (2.3,0) {1};
\node[state] (2) at (1.15,-2) {2};

\path (0) edge node[below] {$b$} (1)
      (1) edge node[left] {$b$} (2);

\draw[->] (0) to [out=170,in=130,looseness=8] node[above,yshift=1mm] {$c,d$} (0);
\draw[->] (1) to [out=50,in=10,looseness=8] node[above,yshift=1mm] {$c$} (1);
\path (2) edge[bend right] node[right] {$c,d$} (1);

\path (1) edge[bend right] node[above] {$d$} (0)
      (2) edge node[left] {$b$} (0);

\node at (1.15,-3) {$\mE_3$};

\begin{scope}[xshift=7cm]

\node[state] (0r) at (0,0) {0};
\node[state] (1r) at (2.3,0) {1};
\node[state] (2r) at (2.3,-2.3) {2};
\node[state] (3r) at (0,-2.3) {3};

\path (0r) edge node[below] {$b$} (1r)
      (1r) edge node[pos=0.3,left] {$b$} (2r)
      (2r) edge node[above] {$b$} (3r)
      (3r) edge node {$b$} (0r);

\draw[->] (0r) to [out=170,in=130,looseness=8] node[above,yshift=1mm] {$c,d$} (0r);
\draw[->] (1r) to [out=50,in=10,looseness=8] node[above,yshift=1mm] {$c$} (1r);

\path (2r) edge[bend right] node[right] {$c$} (1r)
      (3r) edge node[pos=0.2,sloped,above] {$c,d$} (1r);

\path (1r) edge[bend right] node[above] {$d$} (0r)
      (2r) edge node[pos=0.2,sloped,above] {$d$} (0r);

\node at (1.15,-3) {$\mE_4$};

\end{scope}

\end{tikzpicture}
\caption{Automata $\mE_3\in\mathbf{A}_2\setminus\mathbf{A}_3$ and $\mE_4\in\mathbf{A}_3\setminus\mathbf{A}_4$} \label{fig:strict}
\end{center}
\end{figure}

\begin{lemma}\label{lem:strict}
For each $n\ge 3$, the automaton $\mE_n$ lies in $\mathbf{A}_{n-1}\setminus\mathbf{A}_n$.
\end{lemma}

\begin{proof}
We need to show that Alice has a winning strategy in the $(n-1)$-game on $\mE_n$, whereas Bob can win the $n$-game on this automaton.
	
The rules \eqref{eq:rulesEn} give $\{0,n-1\}{\cdot} x=\{0,1\}$ for each letter $x\in\{b,c,d\}$. Hence
\[
\{0,n-1\}{\cdot} xb^{n-1}=\{0,1\}{\cdot} b^{n-1}=\{0+(n-1)\kern-7pt\pmod{n},\,1+(n-1)\kern-7pt\pmod{n}\}=\{0,n-1\}.
\]
Therefore, if in the $n$-game on $\mathcal{E}_n$ Bob plays the word $b^{n-1}$ on every one of his turns, then after each of his moves both states $0$ and $n-1$ carry tokens. Consequently, Bob prevents synchronization and wins the game.

To win the $(n-1)$-game on $\mE_n$, Alice chooses one of the letters $c$ or $d$ on her first turn. After that, only states 0 and 1 hold tokens. Bob responds with an arbitrary word $w$ of length $<n-1$, which moves the tokens to the states in $\{0,1\}{\cdot}w$. If $n-1\notin\{0,1\}{\cdot}w$, Alice wins by playing the letter $d$. The rules \eqref{eq:rulesEn} ensure that $m{\cdot}x\le m+1$ for each $m\ne n-1$ and each $x\in\{b,c,d\}$. From this, the containment $n-1\in\{0,1\}{\cdot}w$ is possible only when $w=b^{n-2}$. In this case, $\{0,1\}{\cdot}w=\{n-2,n-1\}$, and Alice wins by playing the letter $c$.
\end{proof}

\begin{corollary}\label{cor:strict}
All inclusions in\/ \eqref{eq:hierarchy} are strict.
\end{corollary}

\begin{proof}
For $k>1$, the strict containment $\mathbf{A}_k\supsetneqq \mathbf{A}_{k+1}$ is witnessed by Lemma~\ref{lem:strict}. By (O1), the strictness of $\mathbf{A}_1\supsetneqq \mathbf{A}_2$ amounts to saying that Bob can win the 2-synchronization game on a \san. This is witnessed by any \v{C}ern\'{y} automaton $\mathrsfs{C}_{n}$ with $n\ge 3$. Indeed, according to~\cite[Example 1]{FMV}, Bob wins on $\mathrsfs{C}_{n}$ for $n\ge 3$ even if he is not allowed to skip moves, that is, to play the empty word.
\end{proof}

Two questions about the hierarchy \eqref{eq:hierarchy} arise naturally. The first is how to efficiently determine the level of~\eqref{eq:hierarchy} which a given \san{} $\mA$ reaches, that is, to compute the largest $k\in\overline{\mathbb{N}}$ such that $\mA\in\mathbf{A}_k$. The second is how the \rl{} of $\mA$ depends on its level. We address these questions in Sections~\ref{sec:decidability} and~\ref{sec:speed}, respectively.

Yet another question of interest concerns the behaviour of the restriction of the hierarchy~\eqref{eq:hierarchy} to DFAs with a fixed number of states. This is studied in Section~\ref{sec:nstates}.

\subsection{Localization} Sometimes, we will need a localization result that is a slight modification of \cite[Lemma 2]{FMV}. For the reader's convenience we include its proof.

\begin{lemma}
\label{lem:localization} Alice has a winning strategy in the $k$-game on a DFA if and only if she has a winning strategy from every position in which only two states of the DFA hold tokens.
\end{lemma}

\begin{proof}
If Alice has no winning strategy from a position with two tokens, say $T$ and $T'$, then Bob has a winning strategy from this position. If Bob plays from the initial position according to this strategy, that is, choosing his moves solely on the basis of the locations of the tokens $T$ and $T'$, as if no other tokens were present, then the two tokens persist forever, and Alice loses the game. (Here we assume that whenever one of the tokens $T$ and $T'$ meets some third token on some state during the course of the game, then it is this third token that is removed.)

Conversely, if Alice can win from every position in which only two states hold tokens, she may proceed as follows. In the initial position, she selects a pair of tokens, say $T$ and $T'$, and plays as if no other tokens were present, that is, she applies her winning strategy from the position in which $T$ and $T'$ occupy the same states as they do in the initial position and all other tokens are removed. This brings the game to a position in which $T$ or $T'$ is removed. Alice then selects a new pair of tokens and plays as if these were the only tokens, and so on. Since at least one token is removed in each round, Alice eventually wins.
\end{proof}

The proof of Lemma~\ref{lem:localization} implies an observation similar to that in \cite[Corollary 3]{FMV}:

\begin{corollary} \label{cor:cubic}
If Alice has a winning strategy in the $k$-game on a DFA with $n$ states, she can win in at most $\binom{n}2(n-2)+1$ moves.
\end{corollary}

\begin{proof}
Select a pair of tokens $T$ and $T'$ and let $q_i$ and $q'_i$ denote the states holding these tokens $T$ and $T'$ after the $i$-th move of Alice. Then if Alice plays optimally, we must have $\{q_i,q'_i\}\ne\{q_{i+j},q'_{i+j}\}$ for every $j>0$. Indeed, the equality $\{q_i,q'_i\}=\{q_{i+j},q'_{i+j}\}$ would mean that, regardless of how Alice moves $T$ and $T'$ by her $(i+1)$-st, \dots, $(i+j-1)$-st moves, Bob can force Alice to return the tokens, by her $(i+j)$-th move, to the same states they occupied after her $i$-th move. Analogously, Bob can force Alice to return $T$ and $T'$ to the same states also after her $(i+2j)$-th, $(i+3j)$-th, \dots\ moves. Consequently, neither of the two tokens can ever be removed, contradicting the assumptions that Alice has a winning strategy and plays optimally.

Hence the number of Alice's moves in any round in which she operates with any fixed pair of tokens does not exceed $\binom{n}2$. Moreover, in every \san\ there exist states $q$ and $q'$ such that $q{\cdot}a=q'{\cdot}a$ for some letter $a$. Therefore Alice can remove one token by her first move. After that, she needs at most $n-2$ rounds to remove $n-2$ of the remaining $n-1$ tokens. Altogether, Alice needs at most $\binom{n}2(n-2)+1$ moves.
\end{proof}

\section{Synchronization speed}
\label{sec:speed}

\subsection{Reset thresholds of $A_\omega$-automata}
\label{subsec:omegaspeed}
We aim to show that the \rl{} of every $A_\omega$-automaton is less than the number of its states. In a rough approximation, the idea of the proof is that we assign to each position in the $\omega$-game on a DFA
$(Q, \Sigma)$ the value of a certain parameter with the following properties: first, Alice always has a move that decreases the current value (in an appropriate sense), while no move of Bob can increase it; and second, the total number of possible values of this parameter is less than $|Q|$. To embody this idea, we need several concepts and pieces of notation.

For a DFA $\mA = (Q, \Sigma)$, we define a sequence $\{E_\ell\}_{\ell=0,1,2,\dots}$ of binary relations on the set $Q$ by induction:
	\begin{itemize}
		\item $E_0 := \{(q, q) : q \in Q\}$;
		
		\item $E_{\ell+1}$ is the set of all pairs $(p,q)\in Q\times Q$ such that for every word $w \in \Sigma^*$, there exists a letter $x\in\Sigma$ such that $(p{\cdot}wx,\,q{\cdot}wx)\in E_\ell$.
	\end{itemize}
In terms of the $\omega$-game on $\mA$, the definition says that $(p,q)\in E_\ell$ if and only if Alice has a strategy that allows her to move tokens from the states $p$ and $q$ to a common state in at most $\ell$ moves, provided that Bob moves first.

A straightforward induction on $\ell$ shows that $(p,q)\in E_\ell$ if and only if $(q, p)\in E_\ell$ for all $\ell$; that is, the relations $E_\ell$ are symmetric. We therefore regard the pairs $(Q,E_\ell)$ as (undirected) graphs and use the corresponding terminology. In this context, we take the liberty of occasionally referring to pairs $(p,q)\in E_\ell$ as \emph{edges}, although, strictly speaking, the edges of $(Q,E_\ell)$ are the sets $\{p,q\}$ with $(p,q)\in E_\ell$.

We record a few properties of the relations $E_\ell$.

\begin{proposition}\label{prop:PRT1} Let $\mA = (Q, \Sigma)$ be a DFA, $\ell$ a nonnegative integer.
	\begin{description}
		\item[\it(Stability)] if $(p,q)\in E_\ell$, then $(p{\cdot}w, q{\cdot}w) \in E_\ell$ \ for every word $w \in \Sigma^*$;
		
		\item[\it(Containment)] $E_\ell\subseteq E_{\ell+1}$;
		
		\item[\it(Criterion)] $\mA$ is an $A_\omega$-automaton if and only if\/ $\bigcup_{\ell \geq 0} E_\ell = Q\times Q$.
	\end{description}
\end{proposition}

\begin{proof}
\textit{Stability.} The claim is immediate for $\ell=0$. Let $\ell>0$, and denote $p':=p{\cdot}w$ and $q':=q{\cdot}w$. Take an arbitrary word $w'\in\Sigma^*$. Then $p'{\cdot}w'=p{\cdot}ww'$ and $q'{\cdot}w'=q{\cdot}ww'$. Applying the definition of $E_\ell$ to the pair $(p,q)$ and the word $ww'$ ensures the existence of a letter $x\in\Sigma$ such that $(p{\cdot}ww'x,\,q{\cdot}ww'x)\in E_{\ell-1}$. Thus, $(p'{\cdot}w'x,\,q'{\cdot}w'x)\in E_{\ell-1}$, and this verifies that $(p',q')\in E_\ell$.

\smallskip
	
\textit{Containment.} We induct on $\ell$. The base $E_0 \subseteq E_1$ is clear. Let $\ell > 0$. By definition, for each $(p,q) \in E_\ell$ and each $w \in \Sigma^*$, there is a letter $x\in\Sigma$ such that $(p{\cdot}wx,\,q{\cdot}wx)\in E_{\ell-1}$. By the inductive assumption, $E_{\ell-1} \subseteq E_\ell$, whence $(p{\cdot}wx,\,q{\cdot}wx)\in E_\ell$. Thus, $(p,q) \in E_{\ell+1}$.

\smallskip

\textit{Criterion.} The `only if' part. Since the set $Q$ is finite and $E_\ell\subseteq E_{\ell+1}$ for all $\ell\ge 0$, there exists an integer $L$ such that $\bigcup_{\ell \geq 0} E_\ell=E_L$. Let $H:=(Q\times Q)\setminus E_L$ and suppose, towards a contradiction, that $H\ne\varnothing$. Then, for every pair $(p,q) \in H$, we have $p\ne q$, and there exists a word $w_{(p,q)}\in\Sigma^*$ such that $(p{\cdot}w_{(p,q)}x,\,q{\cdot}w_{(p,q)}x)\in H$ for every letter $x\in\Sigma$. Indeed, otherwise, for every $w\in\Sigma^*$, one could find a letter $x\in\Sigma$ such that $(p{\cdot}wx,\,q{\cdot}wx)\in E_L$ which would mean that $(p,q)\in E_{L+1}=E_L$.

Now, Bob wins the $\omega$-game on $\mA$, using the following strategy. He first chooses a pair $(p_0,q_0)\in H$ and plays the word $w_1:=w_{(p_0,q_0)}$.  If Alice then plays a letter $x_1$ and we set $p_1 := p_0{\cdot}w_1x_1$ and $q_1 := q_0{\cdot}w_1x_1$, Bob responds by playing the word $w_2 := w_{(p_1,q_1)}$, and so on. By construction, after each turn the pair of states reached from $(p_0,q_0)$ belongs to $H$, and hence consists of two distinct states. Therefore, at no turn can the history of the game be a reset word for $\mA$. This shows that Bob has a winning strategy on $\mA$, contradicting the assumption that $\mA$ is an $A_\omega$-automaton.

\smallskip

The `if' part. By Lemma~\ref{lem:localization}, it suffices to show that Alice wins from every position with only two tokens. Take any such position, and let $p$ and $q$ be the states holding the tokens. Since $\bigcup_{\ell\ge0} E_\ell = Q\times Q$, there exists $\ell$ such that the pair $(p,q)$ belongs to the relation $E_\ell$. As observed after the definition of $E_\ell$, Alice then has a strategy that allows her to move the tokens from the states $p$ and $q$ to a common state in at most $\ell$ moves.
\end{proof}

By  Proposition \ref{prop:PRT1}, in any $A_\omega$-automaton $(Q,\Sigma)$, every pair $(p,q)\in Q\times Q$ belongs to some relation $E_\ell$. We denote by $d(p,q)$ the smallest integer $\ell$ with this property. All items in the next proposition follow easily from the definitions and Proposition \ref{prop:PRT1}.

\begin{proposition}\label{prop:dproperties}
	In any $A_\omega$-automaton $(Q,\Sigma)$, the following hold for all $p,q \in Q$:
	\begin{enumerate}
		\item[\rm(D1)] $d(p,q) = 0$ if and only if $p = q$;
		
		\item[\rm(D2)] $d(p,q) = d(q,p)$;
		
		\item[\rm(D3)] $d(p{\cdot}w,\,q{\cdot}w) \le d(p,q)$ for every word $w\in\Sigma^*$;
		
		\item[\rm(D4)] if $d(p,q) > 0$, then there is a letter $x\in\Sigma$ such that $d(p{\cdot}x,\,q{\cdot}x)<d(p,q)$.
	\end{enumerate}
\end{proposition}

\begin{remark}
It can be verified that the function $d$ also satisfies the triangle inequality $d(p,r) \le d(p,q) + d(q,r)$, that is, $d$ is a metric on the set $Q$. However, we do not employ this fact.
\end{remark}

We need some notation related to the connected components of the graphs $(Q,E_\ell)$, where $(Q,\Sigma)$ is an $A_\omega$-automaton. For brevity, we refer to the set of all vertices of a connected component of a graph $G$ as a \emph{component} of $G$.

For every integer $\ell\ge 0$, let $\mathcal{D}_\ell$ denote the set of all components of the graph $(Q, E_\ell)$. Each $\mathcal{D}_\ell$ is a partition of the set $Q$, and when viewed as such, $\mathcal{D}_{\ell+1}$ is coarser than $\mathcal{D}_{\ell}$. Thus, every component in $\mathcal{D}_{\ell+1}$ is a union of components from $\mathcal{D}_\ell$. We therefore obtain a descending chain of partitions with respect to the refinement order. The chain starts with the one-block partition $\{Q\}$ (which is $\mathcal{D}_K$ for some integer $K$ such that $E_K = Q\times Q$; such a~$K$ exists since $\bigcup_{\ell \geq 0} E_\ell = Q\times Q$ and the set $Q$ is finite) and ends with $\mathcal{D}_0 = \{\{q\} : q \in Q\}$. We set $\mathcal{D}:= \bigcup_{\ell\ge 0} \mathcal{D}_\ell$ and $\mathcal{D}^+:=\mathcal{D}\setminus\mathcal{D}_0$. Notice that any two components in $\mathcal{D}$ are either disjoint or comparable with respect to inclusion.

For any component $C\in \mathcal{D}$, let $h(C)$ denote the smallest integer $\ell$ such that $C\in\mathcal{D}_\ell$. Notice that $h(C)=0$ if and only if $C\in\mathcal{D}_0$, that is, if and only if $|C|=1$. Each $C\in\mathcal{D}^+$ is partitioned in more than one components from $\mathcal{D}_{h(C)-1}$. We denote by $\mathcal{V}(C)$ the set of components of this partition of $C$. Thus, $\mathcal{V}(C)\subseteq\mathcal{D}_{h(C)-1}$ and $|\mathcal{V}(C)|>1$.

For every nonempty subset $A\subseteq Q$, there is the least (with respect to inclusion) component in $\mathcal{D}$ that contains $A$. We denote this component by $C_A$. Notice that $|C_A|>1$ if and only if $|A|>1$.
Denoting $h(C_A)$ by $\ell$, we introduce a further bundle of notation:
	\begin{itemize}
		\item $\mathcal{T}_A$ is the family of all connected graphs $G = (V,E)$ with $A\subseteq V\subseteq C_A$ and $E\subseteq E_\ell$;
		
		\item $\theta_\ell(G)$ for $G = (V,E)\in\mathcal{T}_A$ is the number of edges $e\in E$ such that $d(e)=\ell$;
		
		\item $\gamma(A): = \min\limits_{G \in \mathcal{T}_A} \theta_\ell(G)$.
	\end{itemize}

\begin{lemma}\label{lem:gamma}
For every $A_\omega$-automaton $(Q,\Sigma)$ and every nonempty subset $A\subseteq Q$,
\[
\gamma(A)\le\gamma(C_A)=\begin{cases}
			|\mathcal{V}(C_A)| - 1 &\text{if }\ |A|>1; \\
			0 &\text{otherwise.}
		    \end{cases}
\]
\end{lemma}

\begin{proof}
To lighten notation, we write $C$ for $C_A$ in this proof.

The inequality $\gamma(A)\le\gamma(C)$ follows from the fact that the minimum in the expression for $\gamma(A)$ is taken over a larger set since clearly $\mathcal{T}_C$ is a subfamily of the family $\mathcal{T}_A$.

To prove the equality
\begin{equation}\label{eq:gamma}
\gamma(C)= \begin{cases}
			|\mathcal{V}(C)| - 1 &\text{if } C\in\mathcal{D}^+; \\
			0 &\text{otherwise},
		    \end{cases}
\end{equation}
first observe that $C_C=C$; hence all graphs in the family $\mathcal{T}_C$ have $C$ as their vertex set.

If $|C|=1$, then the one-vertex graph with no edges occurs in the family $\mathcal{T}_C$, whence $\min\limits_{G\in\mathcal{T}_C} \theta_\ell(G)=0$, which agrees with~\eqref{eq:gamma}.

Recall that each set $C\in\mathcal{D}^+$ is partitioned into $|\mathcal{V}(C)|$ components from $\mathcal{D}_{\ell-1}$. For a graph $G=(C,E)$ with $E\subseteq E_\ell$ to be connected, any two such components must be connected by a path involving edges in $E_\ell\setminus E_{\ell-1}$; see Figure~\ref{fig:partition} for illustration.
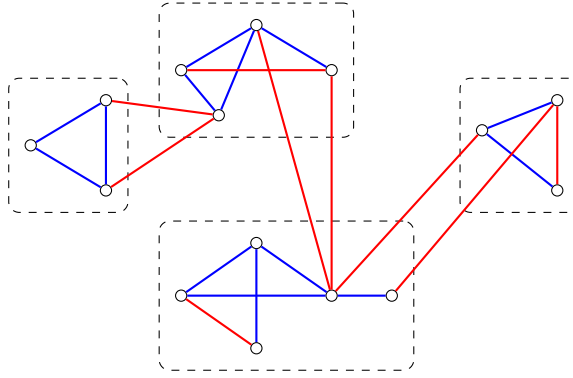
\begin{figure}[hbt]
\begin{center}
\begin{tikzpicture}[
  vertex/.style={circle,draw,fill=white,inner sep=1.5pt},
  inedge/.style={blue,thick},
  outedge/.style={red,thick},
  contour/.style={draw,dashed,rounded corners,inner sep=6pt}
]

\node[vertex] (a1) at (0,0) {};
\node[vertex] (a2) at (1,0.6) {};
\node[vertex] (a3) at (1,-0.6) {};

\draw[inedge] (a1)--(a2);
\draw[inedge] (a1)--(a3);
\draw[inedge] (a2)--(a3);

\node[contour,fit=(a1)(a2)(a3)] {};

\node[vertex] (b1) at (2,1) {};
\node[vertex] (b2) at (3,1.6) {};
\node[vertex] (b3) at (2.5,0.4) {};
\node[vertex] (b4) at (4,1) {};

\draw[inedge] (b1)--(b2);
\draw[inedge] (b2)--(b4);
\draw[inedge] (b1)--(b3);
\draw[inedge] (b3)--(b2);

\node[contour,fit=(b1)(b2)(b3)(b4)] {};

\node[vertex] (c1) at (2,-2) {};
\node[vertex] (c2) at (3,-1.3) {};
\node[vertex] (c3) at (4,-2) {};
\node[vertex] (c4) at (3,-2.7) {};
\node[vertex] (c5) at (4.8,-2) {};

\draw[inedge] (c1)--(c2);
\draw[inedge] (c1)--(c3);
\draw[inedge] (c2)--(c3);
\draw[inedge] (c2)--(c4);
\draw[inedge] (c3)--(c5);

\node[contour,fit=(c1)(c2)(c3)(c4)(c5)] {};

\node[vertex] (d1) at (6,0.2) {};
\node[vertex] (d2) at (7,0.6) {};
\node[vertex] (d3) at (7,-0.6) {};

\draw[inedge] (d1)--(d2);
\draw[inedge] (d1)--(d3);

\node[contour,fit=(d1)(d2)(d3)] {};

\draw[outedge] (a2)--(b3);
\draw[outedge] (a3)--(b3);

\draw[outedge] (b2)--(c3);
\draw[outedge] (b4)--(c3);
\draw[outedge] (b1)--(b4);

\draw[outedge] (c4)--(c1);
\draw[outedge] (c3)--(d1);
\draw[outedge] (c5)--(d2);
\draw[outedge] (d3)--(d2);

\end{tikzpicture}
\caption{A graph from $\mathcal{T}_C$ with components from $\mathcal{V}(C)$ shown by dashed contours. Blue edges belong to $E_{\ell-1}$ (loops from $E_0$ are not shown); red edges are from $E_\ell\setminus E_{\ell-1}$.} \label{fig:partition}
\end{center}
\end{figure}

For each edge $e\in E_\ell\setminus E_{\ell-1}$, we have $d(e) = \ell$ by the definition of the function $d$, and at least $|\mathcal{V}(C)|-1$ edges are required to connect $|\mathcal{V}(C)|$ components. Hence, $\theta_\ell(G)\ge |\mathcal{V}(C)|-1$ for every graph $G\in\mathcal{T}_C$. On the other hand, $|\mathcal{V}(C)|-1$ edges suffice. To see this, consider the graph $H:=(\mathcal{V}(C),F)$ in which two vertices $V_1,V_2\in\mathcal{V}(C)$ are connected by an edge $f$ if and only if there exists an edge $e\in E_\ell\setminus E_{\ell-1}$ connecting $V_1$ and $V_2$ in the graph $(C,(C\times C)\cap E_\ell)$; we say that $e$ \emph{induces}~$f$. The graph $H$ is connected, and any spanning tree of $H$ has $|\mathcal{V}(C)|-1$ edges. Take a spanning tree $H'$ of $H$ and consider the graph $G'=(C,((C\times C)\cap E_{\ell-1})\cup E')$, where $E'$ consists of $|\mathcal{V}(C)|-1$ edges in $E_\ell\setminus E_{\ell-1}$ that induce the edges of the tree $H'$. The graph $G'$ is connected and so it belongs to the family $\mathcal{T}_C$. By construction,  $\theta_\ell(G')=|\mathcal{V}(C)|-1$. Thus, $\min\limits_{G \in \mathcal{T}_C} \theta_\ell(G)=|\mathcal{V}(C)|-1$, which again agrees with~\eqref{eq:gamma}.
\end{proof}
	
For a nonempty subset $A\subseteq Q$, the pair $\chi(A):= (C_A, \gamma(A))$ is called the \emph{characteristic} of $A$. The characteristics of singletons have the form $(\{q\},0)$ for some $q\in Q$; we call such characteristics
\emph{trivial}. For every nontrivial characteristic $(C_A,\gamma(A))$, we have $|C_A|>1$ and $1\le\gamma(A)\le\gamma(C_A)$. The next two propositions show that nontrivial characteristics are suited to play the role of the `parameter' from the proof outline at the beginning of this section.

\begin{proposition}\label{prop:PRT4}
For every $A_\omega$-automaton $(Q,\Sigma)$, there are at most $|Q|-1$ distinct nontrivial characteristics.
\end{proposition}

\begin{proof}
For every component $C\in\mathcal{D}^+$, there are at most $\gamma(C)$ distinct characteristics with first entry $C$ because the possible values of the second entry belong to the set $\{1,\dots,\gamma(C)\}$. Summing over all such components yields at most $\sum\limits_{C \in\mathcal{D}^+}\gamma(C)$ distinct nontrivial characteristics. Clearly, $\sum\limits_{C \in\mathcal{D}^+}\gamma(C)=\sum\limits_{C \in\mathcal{D}}\gamma(C)$ as extending the summation to components from $\mathcal{D}\setminus\mathcal{D}^+=\mathcal{D}^0$ adds only zero summands. Let us compute the latter sum.

For every component $R \in \mathcal{D}$, we set
\begin{equation}\label{eq:F}
\Gamma(R):= \sum\limits_{\substack{C \in \mathcal{D} \\ C \subseteq R}} \gamma(C).
\end{equation}
We prove that $\Gamma(R) = |R| - 1$ by induction on $|R|$. The equality is immediate if $|R|=1$. Suppose that $|R|>1$. Then $|\mathcal{V}(R)|>1$, whence the inductive assumption applies to every component in $\mathcal{V}(R)$.
Clearly, for every $C\in\mathcal{D}$ such that $C\subsetneqq R$, there is a unique component $S\in\mathcal{V}(R)$ such that $C\subseteq S$. We then have
\begin{align*}
  \Gamma(R)&= \gamma(R) + \sum\limits_{\substack{C \in \mathcal{D} \\ C \subsetneqq R}} \gamma(C)
      = \gamma(R) + \sum\limits_{S \in \mathcal{V}(R)}\sum\limits_{\substack{C \in \mathcal{D} \\ C \subseteq S}} \gamma(C)\\
      &= \gamma(R) + \sum\limits_{S \in \mathcal{V}(R)}\Gamma(S)&&\text{by \eqref{eq:F}}\\
      &=\gamma(R) + \sum\limits_{S \in \mathcal{V}(R)}(|S| - 1)&&\text{by the inductive assumption}\\
      &=|\mathcal{V}(R)| - 1 + \sum\limits_{S \in \mathcal{V}(R)}(|S| - 1)&&\text{by \eqref{eq:gamma}}\\
      &=|\mathcal{V}(R)| - 1 + \sum\limits_{S \in \mathcal{V}(R)}|S| - |\mathcal{V}(R)|= |R| - 1&&\text{since $|R|=\sum\limits_{S \in \mathcal{V}(R)}|S|$.}
\end{align*}
In particular, $\Gamma(Q) = \sum\limits_{C \in \mathcal{D}} \gamma(C) = |Q| - 1$. This establishes the claimed bound on the number of different nontrivial characteristics.
\end{proof}

We say that a nonempty subset $A\subseteq Q$ \emph{avoids} a characteristic $\mu$ if $\mu\ne\chi(A{\cdot}w)$ for every word $w\in\Sigma^*$.

\begin{proposition}\label{prop:PRT5}
	For every $A_\omega$-automaton $(Q,\Sigma)$ and every subset $A\subseteq Q$ with $|A|>1$, there exists a letter $x\in \Sigma$ such that the set $A{\cdot}x$ avoids the characteristic $\chi(A)$.
\end{proposition}

\begin{proof}
Denote $\ell = h(C_A)$ and choose a graph $G =(V,E)\in\mathcal{T}_A$ such that $\theta_\ell(G) = \gamma(A)$. Recall that by the definition of the family $\mathcal{T}_A$, the graph $G$ is connected and $E\subseteq E_\ell$. Choose an edge $(p_0,q_0)\in E$ such that $d(p_0,q_0) = \ell$. As registered in Proposition \ref{prop:dproperties}, item (D4), there exists a letter $x\in\Sigma$ such that $d(p_0{\cdot}x,\,q_0{\cdot}x)<\ell$. We now show that $A{\cdot}x$ avoids the characteristic $\chi(A)$.
	
Take an arbitrary word $w\in\Sigma^*$ and set $B:= A{\cdot}xw$. We aim to prove that $\chi(B)\ne\chi(A)$. This is immediate if $C_A\ne C_B$, so assume that $C_A=C_B$.

Consider the image $G'$ of the graph~$G$ under the action of the word $xw$, that is, $G':= (V',E')$, where $V':= V{\cdot}xw$ and $E':= \{(p{\cdot}xw,\,q{\cdot}xw) : (p,q)\in E\}$. We aim to prove that the graph $G'$ lies in the family $\mathcal{T}_B$, which amounts to verifying the following:
\begin{enumerate}
  \item[(i)] $E'\subseteq E_{h(C_B)}$;
  \item[(ii)] $G'$ is connected;
  \item[(iii)] $B\subseteq V'\subseteq C_B$.
\end{enumerate}

The inclusion (i) readily follows from the equality $h(C_B)=h(C_A)=\ell$ (ensured by the equality $C_A=C_B$), the inclusion $E\subseteq E_\ell$, and the stability of $E_\ell$ shown in Proposition~\ref{prop:PRT1}.

For (ii), it suffices to notice that if a sequence of consecutive edges
\[
(p_1,p_2),(p_2,p_3),\dots,(p_{m-1},p_m)\in E
\]
connects two states $p=p_1$ and $p_m=q$ from $A$, then, after removing from its image
 \[
(p_1{\cdot}xw,\,p_2{\cdot}xw),(p_2{\cdot}xw,\,p_3{\cdot}xw),\dots,(p_{m-1}{\cdot}xw\,,p_m{\cdot}xw)
 \]
the pairs belonging to $E_0$, if any, we obtain a sequence of consecutive edges in $E'$ connecting the states $p{\cdot}xw$ and $q{\cdot}xw$ from $B$.

A similar argument can be used to justify the inclusion $V'\subseteq C_B$ in (iii). Indeed, every state in $V'$ is of the form $q{\cdot}xw$ for some $q\in V$. Since $V\subseteq C_A$ and $C_A$ is the component of the graph $(Q,E_\ell)$ containing $A$, the state $q$ is connected by a sequence of consecutive edges in $E_\ell$ with a state $p\in A$. Then the image of this sequence connects $q{\cdot}xw$ with $p{\cdot}xw\in B$ and each edge in the image belongs to $E_\ell$ due to the stability of $E_\ell$. Hence, $q{\cdot}xw$ lies in the component of $(Q,E_\ell)$ containing $B$, that is, $q{\cdot}xw\in C_B$.

The inclusion $B\subseteq V'$ in (iii) is ensured by the inclusion $A\subseteq V$.

Thus, $G'$ is a graph in $\mathcal{T}_B$. By item (D3) of Proposition \ref{prop:dproperties}, $d(p{\cdot}xw,\,q{\cdot}xw)\le\ell$ for every edge $(p,q)\in E_\ell$. By the same property, for the edge $(p_0,q_0)$, this inequality is strict. Hence, $\theta_\ell(G') < \gamma(A)$, so $\gamma(B) < \gamma(A)$. Hence, $\chi(B)\ne\chi(A)$.
\end{proof}

We are ready for the proof of our main result.

\begin{theorem}
\label{thm:main}
	On every $A_\omega$-automaton $\mA=(Q,\Sigma)$, Alice has a strategy that allows her to win in less than $|Q|$ moves. In particular, the reset threshold of $\mA$ is less than $|Q|$.
\end{theorem}

\begin{proof}
We may assume that Bob moves first. Denote by $w_i$ the history up to Bob's $i$-th turn, and let $A_i:=Q{\cdot}w_i$. Alice's strategy is to respond with a letter $x_i$ such that the set $A_i{\cdot}x_i$ avoids the characteristic $\chi(A_i)$. By Proposition~\ref{prop:PRT5}, such a letter $x_i$ exists as long as $|A_i|>1$, that is, as long as $w_i$ is not a reset word for $\mA$.
	
The described strategy ensures that $\chi(A_i)\ne\chi(A_j)$ whenever $i\ne j$. Indeed, suppose that $i<j$; then $w_j = w_ix_iu$ for some word $u\in\Sigma^*$, and consequently $A_j=A_i{\cdot}x_iu$. By construction, $A_i{\cdot}x_i$ avoids $\chi(A_i)$, which by the definition of avoidability means that $\chi(A_i)\ne\chi(A_i{\cdot}xw)$ for every word $w\in\Sigma^*$. Taking $w=u$ yields $\chi(A_i)\ne\chi(A_j)$.  Thus, all characteristics $\chi(A_1),\chi(A_2),\dots$ are distinct.

By Proposition~\ref{prop:PRT5}, there are at most $|Q| - 1$ nontrivial characteristics. Hence, using the described strategy, Alice makes at most $|Q| - 1$ moves in order to win.

The history of the game up to its final move is a reset word for $\mA$, independent of which words Bob has chosen on his moves. In particular, the final history of the game in which Bob has chosen the empty word on each of his moves is a reset word whose length equals the number of Alice's moves and thus does not exceed $|Q| - 1$.
\end{proof}
	
\begin{remarks}
1. It is easy to see that the bound in Theorem~\ref{thm:main} is tight. As a simple witness, for every $n\ge3$, consider the DFA $(\mathbb{Z}_n,\{a\})$ in which $0{\cdot} a:=0$ and $m{\cdot} a:=m-1$ for $0<m<n$. This is an $A_\omega$-automaton with $n$ states and \rl{} $n-1$, and Alice needs exactly $n-1$ moves to win on it.

2. The converse of the second conclusion of Theorem~\ref{thm:main} does not hold. As a witness, for every $n\ge4$, consider the DFA $(\mathbb{Z}_n,,\{a,b\})$ in which the letters move the states in the opposite directions: $0{\cdot} a:=0$ and $m{\cdot}a:=m-1$ for $0<m<n$, whereas $(n-1){\cdot} b:=n-1$ and $m{\cdot}b:=m+1$ for $0\le m<n-1$. The DFA is synchronizing with \rl{} $n-1$ since each of $a^{n-1}$ and $b^{n-1}$ is a reset word of minimum length. On the other hand, Bob wins even the $2$-game on this DFA. Bob's winning strategy is to play $b$ whenever Alice plays $a$, and to play $a$ if she plays $b$. In this way, he can ensure that each of the `middle' states 1, 2, \dots, $n-2$ holds a token after his responses; see Figure~\ref{fig:Bobwins}.
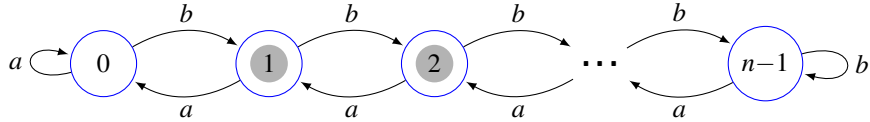
\begin{figure}[htb]
\begin{center}
\begin{tikzpicture}[->,>=latex,shorten >=1pt,auto,node distance=2.2cm,
                    every state/.style={circle, draw=blue, scale=1},
                    token/.style={circle,fill=gray!60,minimum size=5mm,inner sep=0pt},]

  \node[state] (0) {$0$};
  \node[state] (1) [right of=0] {};
  \node[token]  (t1) [right of=0] {$1$};
  \node[state] (2) [right of=1] {};
  \node[token]  (t2) [right of=1] {$2$};
  \node (dots)  [right of=2] {\raisebox{4pt}{$\mathbf{\centerdot\,\centerdot\,\centerdot}$}};
  \node[state] (n1) [right of=dots] {$n{-}1$};

  \path
    (1) edge[bend left] node[below] {$a$} (0)
    (2) edge[bend left] node[below] {$a$} (1)
    (n1) edge[bend left] node[below] {$a$} (dots)
    (dots) edge[bend left] node[below] {$a$} (2);
  \path
    (0) edge[bend left] node[above] {$b$} (1)
    (1) edge[bend left] node[above] {$b$} (2)
    (2) edge[bend left] node[above] {$b$} (dots)
    (dots) edge[bend left] node[above] {$b$} (n1);

  \path
    (0) edge[loop left] node {$a$} ()
    (n1) edge[loop right] node {$b$} ();

\end{tikzpicture}
\end{center}
\caption{An $n$-state DFA with \rl{} $n-1$ on which Bob wins by keeping tokens (shown in gray) on states 1, 2, \dots, $n-2$.}
\label{fig:Bobwins}
\end{figure}
\end{remarks}

\subsection{Reset threshold of DS-automata}
\label{subsec:dsautomata}

Theorem~\ref{thm:main} confirms the natural expectation that automata on which Alice wins against a powerful adversary should be quickly synchronizable. However, the property of being an $A_\omega$-automaton may appear too strong, and it is far from clear whether Theorem~\ref{thm:main} applies to any meaningful class of DFAs not defined in terms of synchronization games. Thanks to an observation made in \cite[Section 5.4]{JALC}, we are in a position to exhibit such a class. To introduce it, we recall the notion of the transition monoid of a DFA and a few basic concepts from semigroup theory.

For a DFA $\mathrsfs{A}=(Q,\Sigma)$, the map $\tau_a\colon Q\to Q$ defined by $q\tau_a:=q{\cdot}a$ is a transformation of the set $Q$. The \emph{transition monoid} of a DFA $\mA=(Q,\Sigma)$ is the submonoid of the monoid of all transformations on the set $Q$ generated by the set $\{\tau_a : a\in\Sigma\}$. We denote the transition monoid of a DFA $\mA$ by $T(\mA)$. In $T(\mA)$, any product of the form $\tau_{a_1}\tau_{a_2}\cdots\tau_{a_n}$ is precisely the transformation $\tau_w\colon Q\to Q$ defined by $q\tau_w:=q{\cdot}w$, where $w:=a_1a_2\cdots a_n$. Thus, the transition monoid $T(\mA)$ can equivalently be defined as the monoid of all transformations on the set $Q$ induced by the action of words over $\Sigma$.

If $\mathrsfs{A}=(Q,\Sigma)$ is a synchronizing DFA and $w$ is a reset word for $\mA$, then the transformation $\tau_w$ is a constant map on $Q$, that is, $Q\tau_w=\{q\}$ for some $q\in Q$. Thus, the transition monoid of a \san{} always contains a constant transformation. Conversely, if $\zeta\in T(\mA)$ is a constant transformation, then any word $w$ with $\tau_w=\zeta$ is a reset word for $\mA$, and hence $\mA$ is synchronizing. We therefore see that synchronization is actually a property of the transition monoid of an automaton rather than the automaton itself: for DFAs $\mathrsfs{A}=(Q,\Sigma)$ and $\mathrsfs{A}'=(Q,\Sigma')$ with the same states but different input alphabets, the equality $T(\mA)=T(\mA')$ guarantees that $\mA'$ is synchronizing if and only if $\mA$ is.

Recall the definition of \emph{Green's relations}  $\gR$, $\gL$, and $\gD$ on an arbitrary semigroup $S$:
\begin{itemize}
\item $a\gR b \Longleftrightarrow {}$ either $a=b$ or $a=bs$ and $b=at$ for some $s,t\in S$;
\item $a\gL \,b \Longleftrightarrow {}$ either $a=b$ or $a=sb$ and $b=ta$ for some $s,t\in S$;
\item $a\gD b \Longleftrightarrow {}$ $a\gR c$ and $c\gL b$ for some $c\in S$.
\end{itemize}
The relations $\gR$ and $\gL$ are clearly equivalences. The definition of ${\gD}$ means that ${\gD}$ is the product of $\gR$ and $\gL$ as binary relations. As observed in \cite{Green:1951}, ${\gD}$ is also the product of $\gL$ and $\gR$, which implies that $\gD$ is the least equivalence containing both $\gR$ and $\gL$.

A $\gD$-class $D$ is called \emph{regular} if it contains an element $a$ such that $aba=a$ for some $b\in D$. DFAs in whose transition monoids all regular $\gD$-classes are subsemigroups often appear in the literature; see, e.g., \cite{AMSV09,AlSt09}. In~\cite{NCMA}, such DFAs were coined \emph{DS-automata}, and we adopt this terminology here. We refer the reader to~\cite[Section 4.1]{NCMA} for an overview of studies of \sa{} in various subclasses of DS-automata.

The main result of~\cite{JALC} states that, on every synchronizing DS-automaton, Alice has a \emph{uniform} winning strategy in the synchronization game in which Bob's responses are restricted to letters. Here, uniformity means that Alice's winning sequence of moves depends only on the underlying automaton and does not depend on Bob's responses. In \cite[Section 5.4]{JALC}, this strategy is shown to persist for the $\omega$-game. Hence,  synchronizing DS-automata are $A_\omega$-automata, and Theorem~\ref{thm:main} yields the following:
\begin{corollary} \label{cor:ds}
The reset threshold of every synchronizing DS-automaton is less than the number of its states.
\end{corollary}

The result of Corollary~\ref{cor:ds} is contained in \cite[Theorem 2.6]{AlSt09}, where its proof was based on the classical theory of linear representations of semigroups (the Munn--Ponizovsky theory) and its enhancement in~\cite{AMSV09}. Our proof is quite different, as it is purely combinatorial, apart from invoking the results from~\cite{JALC}, which themselves use only a little bit of algebra.

Observe that an $A_\omega$-automaton need not be a DS-automaton. The simplest witness here is the automaton $\mB_2$ shown in Figure~\ref{fig:autB21}.
\begin{figure}[htb]
\begin{center}
\begin{tikzpicture}
	\node[fill=white, circle, draw=blue, scale=1] (0) {$0$};
	\node[fill=white, circle, draw=blue, scale=1, above = 0, xshift= -2cm, yshift= 0.5cm] (1) {$1$};
    \node[fill=white, circle, draw=blue, scale=1, above = 0, xshift= 2cm, yshift= 0.5cm] (2) {$2$};
	\draw
		(1) edge[-latex, below]  node{$b$} (0)
        (1) edge[-latex, bend left, above]  node{$a$} (2)
		(2) edge[-latex, above] node{$b$}(1)
		(2) edge[-latex, below] node{$a$} (0)
		(0) edge[-latex, loop below, below, in =-65, out = -120, distance = 30] node{$a,b$} (0);
		\end{tikzpicture}
\end{center}
\caption{An $A_\omega$-automaton which is not a DS-automaton}
\label{fig:autB21}
\end{figure}
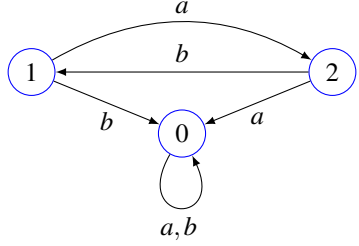
Its transition monoid $T(\mB_2)$ is easy to compute. It consists of six transformations of $\{0,1,2\}$: the identity transformation $\tau_{\varepsilon}$, the constant transformation $\tau_{a^2}$, and the transformations $\tau_a$, $\tau_b$, $\tau_{ab}$, and $\tau_{ba}$ that form a $\gD$-class in $T(\mB_2)$. This $\gD$-class is not a subsemigroup, since, for instance, it does not contain the transformation $\tau_a^2=\tau_{a^2}$. Thus, $\mB_2$ is not a DS-automaton. On the other hand, Alice wins the $\omega$-game on $\mB_2$. She may start with playing $a$, leaving tokens on states 0 and 2. After any non-losing response of Bob, tokens either remain on 0 and 2 or they move to states 0 and~1. Alice then wins by playing $a$ in the former case or $b$ in the latter.

Hence, Theorem~\ref{thm:main} not only provides an alternative proof of Corollary~\ref{cor:ds} but also extends its conclusion to the strictly larger class $\mathbf{A}_\omega$.

It is perhaps worth mentioning that $\mathbf{A}_\omega$ cannot be characterized solely in terms of transition monoids. In fact, for every \sDFA{} $\mathrsfs{A}=(Q,\Sigma)$, there exists an $A_\omega$-automaton $\mathrsfs{A}'=(Q,\Sigma')$  such that $T(\mA)=T(\mA')$. To see this, let $\Sigma':=\Sigma\cup\{c\}$ where the action of the added letter $c$ coincides with the action of a fixed reset word for $\mA$. The transformations induced by the letters in $\Sigma'$ generate the same submonoid of the monoid of all transformations of the set $Q$ as do the transformations induced by the letters in $\Sigma$. Still, Alice wins the $\omega$-game on $\mA'$ instantly by choosing the letter $c$ on her first move.

That said, the question of which properties of the monoid $T(\mA)$ ensure that $\mA$ is an $A_\omega$-automaton appears to be worth further exploration.

\subsection{Iterations of \sa}
\label{subsec:iteration}
What happens if Alice, like Bob, is allowed to play words rather than single letters in an $\omega$-game? Of course, if Alice may play arbitrary words, the game trivializes as Alice can instantly win on any \san{} by spelling out a reset word. The variant in which Alice's moves are subject to a length restriction is more interesting. To handle it, the following construction is well suited.

For $m\in\mathbb{N}$ and an alphabet $\Sigma$, let $\Sigma^{\le m}$ denote the set of all nonempty words over $\Sigma$ of length at most $m$. The $m$-\emph{th iteration} of a DFA $\mA=(Q,\Sigma)$ is the DFA $\mA^{(m)}:=(Q,\Sigma^{\le m})$, with the action of the `letters' in $\Sigma^{\le m}$ defined in a natural way: $q{\circ}w:=q{\cdot}w$ for all $q\in Q$ and $w\in\Sigma^{\le m}$, with the right-hand side denoting the action in $\mA$.

By an $m/\omega$-game we mean the variant in which Alice's moves are nonempty words of length at most $m$, whereas Bob may play arbitrary words. It is then immediate that Alice wins the $m/\omega$-game on a DFA $\mA$ if and only if she wins the standard $\omega$-game played on the DFA $\mA^{(m)}$. It follows that the \rl{} of any $n$-state \san{} on which Alice has a winning strategy in the $m/\omega$-game does not exceed $m(n-1)$. We now show that this upper bound can be slightly improved.

We need the functions $d$, $\mathcal{V}$, $h$, $\theta$, $\gamma$, $\chi$, $A\mapsto C_A$, and $A\mapsto \mathcal{T}_A$ introduced in Subsection~\ref{subsec:omegaspeed}. Note that they all depend on the automaton under consideration, and here we consider two automata: $\mA$ and $\mA^{(m)}$. To avoid any ambiguity, we stipulate that in the following two statements these functions and the sign $\cdot$ are always associated with $\mA^{(m)}$.

\begin{lemma}
\label{lem:first_move}
	If Alice can win the $m/\omega$-game on a DFA $\mA = (Q, \Sigma)$, then there exists a letter $x \in \Sigma$ such that the set $Q{\cdot}x$ avoids the characteristic $\chi(Q)$ in $\mA^{(m)}$.
\end{lemma}

\begin{proof}
Since $|Q| > 1$, the characteristic $\chi(Q)$ is nontrivial and $|\mathcal{V}(Q)| > 1$. For each $q \in Q$ we denote by $[q]$ the only component $C \in \mathcal{V}(Q)$ such that $q \in C$.
	
There exist states $p_0,q_0 \in Q$ and a letter $x \in \Sigma$ such that $[p_0] \ne [q_0]$ and $[p_0{\cdot} x] = [q_0{\cdot}x]$. Indeed, take arbitrary states $p,q$ in $Q$ such that $[p] \ne [q]$ and any reset word $w$ of the DFA $\mA$. Since $p{\cdot} w = q{\cdot} w$, the longest prefix $w'$ of $w$ such that $[p{\cdot}w'] \ne [q{\cdot}w']$ is not equal to $w$, and so it is followed in $w$ by a letter $x$. By the choice of $w'$, we have $[p{\cdot}w'x] = [q{\cdot}w'x]$. Then the states $p_0: = p{\cdot} w'$, $q_0: = q{\cdot} w'$, and the letter $x$ are as required.
	
It is the set $Q{\cdot}x$ that avoids the characteristic $\chi(Q)$. To prove this, we show that $\chi(Q{\cdot}xu) \ne \chi(Q)$ for each word $u\in\Sigma^*$. By definition $\chi(Q{\cdot}xu)=(C_{Q{\cdot}xu},\gamma(Q{\cdot}xu))$ and $\chi(Q)=(C_Q,\gamma(Q))$. If $C_{Q{\cdot}xu}\ne C_Q$, the inequality $\chi(Q{\cdot}xu) \ne \chi(Q)$ is immediate. Therefore we may consider only words $u$ such that $C_{Q{\cdot}xu} = C_Q$. Clearly, $C_Q = Q$, whence $C_{Q{\cdot}xu} = Q$ for any such word $u$. We aim to prove that $\gamma(Q{\cdot}xu) < \gamma(Q)$.

Let $\ell:=h(Q)$ and fix a graph $G_0 = (Q, E) \in \mathcal{T}_Q$ such that $\theta_\ell(G_0) = \gamma(Q)= \min\limits_{G \in \mathcal{T}_Q} \theta_\ell(G)$. Since the graph $G_0$ is connected, it contains a path $\pi$ connecting $p_0$ and $q_0$. By the choice of the states $p_0$ and $q_0$, the components $[p_0]$ and $[q_0]$ are different, whence the path $\pi$ must include an edge $e_0\in E_\ell\setminus E_{\ell-1}$. Let $G_1$ be the graph obtained from $G_0$ by removing the edge $e_0$. By construction $G_1$ has fewer than $\gamma(Q)$ edges $e$ with $d(e)=\ell$.

Consider the image $G'_1$ of the graph~$G_1$ under the action of the word $xu$, that is, $G'_1:= (V',E')$, where $V':= Q{\cdot}xu$ and
\[
E':= \left\{(p{\cdot}xu,\,q{\cdot}xu) : (p,q)\in E\setminus\{e_0\}\right\}.
\]
By item (D3) of Proposition \ref{prop:dproperties}, $d(p{\cdot}xu,\,q{\cdot}xu)\le\ell$ for every edge $(p,q)\in E_\ell$, so the graph $G'_1$ also has fewer than $\gamma(Q)$ edges $e'$ with $d(e')=\ell$. Since $[p_0{\cdot}xu] = [q_0{\cdot}xu]$, the graph $([p_0{\cdot}xu],E_{\ell-1})$, which is connected by the definition of a component, contains a path $\pi'$ connecting $p_0{\cdot}xu$ and $q_0{\cdot}xu$.

Let $G'_0$ be the graph obtained from $G'_1$ by adding all edges of the path $\pi'$ to the set $E'$. We claim that this graph belongs to the family $\mathcal{T}_{Q{\cdot}xu}$, that is,
\begin{enumerate}
  \item[(i)] the edge set of $G'_0$ is contained in $E_{h(Q{\cdot}xu)}$;
  \item[(ii)] $G'_0$ is connected;
  \item[(iii)] $Q{\cdot}xu\subseteq V'\subseteq C_{Q{\cdot}xu}$.
\end{enumerate}

The inclusions (iii) are immediate since $V'=Q{\cdot}xu$ by construction and $C_{Q{\cdot}xu} = Q$ for the word $u$ under consideration. The latter equality also implies that $h(Q{\cdot}xu)=h(Q)=\ell$. Thus, $E_{h(Q{\cdot}xu)}=E_\ell$ and the inclusion $E'\subseteq E_\ell$ follows from the inclusion $E\subseteq E_\ell$ and the stability of $E_\ell$ (Proposition~\ref{prop:PRT1}). Besides $E'$, the edge set of $G'_0$ includes also the edges of the path $\pi'$, but by the choice of $\pi'$, these edges all lie in $E_{\ell-1}\subseteq E_\ell$. This verifies (i).

It remains to verify that the graph $G'_0$ is connected. Recall that the graph $G_0$ was connected and $G_1$ was obtained from $G_0$ by removing a single edge. The graph $G_1$ cannot be connected; otherwise it would belong to $\mathcal{T}_Q$  and have fewer than $\gamma(Q)$ edges in $E_\ell\setminus E_{\ell-1}$, contradicting the definition of $\gamma(Q)$. Thus, $G_1$ has exactly two connected parts, one containing $p_0$ and the other containing $q_0$. The argument used in the proof of Proposition~\ref{prop:PRT5} to verify the connectivity of the graph $G'$ shows that the images of these parts under the action of $xu$ form two connected subgraphs of the graph $G'_1$, containing $p_0{\cdot}xu$ and $q_0{\cdot}xu$, respectively. Since the path $\pi'$ connects $p_0{\cdot}xu$ and $q_0{\cdot}xu$, adding its edges to those of $G'_1$ produces a connected graph.

Thus, we have shown that $G'_0\in\mathcal{T}_{Q{\cdot}xu}$. Then $\gamma(Q{\cdot}xu)\le\theta_{\ell}(G'_0)$ by the definition of the function $\gamma$. On the other hand, $\theta_{\ell}(G'_0)<\gamma(Q)$ since in the edge set of the graph $G'_0$, the set $E'$ has fewer than $\gamma(Q)$ edges in $E_\ell\setminus E_{\ell-1}$ and the edges of the path $\pi'$ all belong to $E_{\ell-1}$. Combining the two inequalities yields $\gamma(Q{\cdot}xu) < \gamma(Q)$, whence $\chi(Q{\cdot}xu) \ne \chi(Q)$.
\end{proof}

\begin{proposition}
\label{prop:iteration}
The \rl{} of an $n$-state automaton on which Alice has a winning strategy in the $m/\omega$-game does not exceed $m(n-2)+1$.
\end{proposition}
\begin{proof}
As observed, if Alice has a winning strategy in the $m/\omega$-game on a DFA $\mA=(Q,\Sigma)$ with $|Q|=n$, then she also has a winning strategy in the $\omega$-game on the $m$-th iteration $\mA^{(m)}=(Q,\Sigma^{\le m})$ of $\mA$. Thus, to win in $n-1$ moves, Alice may use the same strategy as in the proof of Theorem~\ref{thm:main}. Namely, on her $i$-th turn she plays a word $u_i \in \Sigma^{\leq m}$ such that the set $A_i{\cdot}u_i$ avoids the characteristic $\chi(A_i)$, where $A_i$ denotes the position of the game before that turn. By Lemma \ref{lem:first_move}, Alice can choose the first word $u_1$ to consist of a single letter.

If Alice uses the strategy just described and Bob always responds with the empty word, then the concatenation $u_1 \cdots u_{n-1}$ is a reset word for $\mA$ of length at most $m(n - 2) + 1$.
\end{proof}

\begin{remark}
For $m=2$, the bound $2(n-2)+1=2n-3$ provided by Proposition~\ref{prop:iteration} is tight. To see this, for each $n\ge3$, consider the flower automaton $\mF_n: = (\mathbb{Z}_n, \{a_1,\dots,a_{n-1}\})$, in which the letters act as follows:
\[
	m{\cdot} a_1:=\begin{cases}
		0 &\text{if $m = 1$}, \\
		m &\text{if $m \ne 1$};
	\end{cases}\qquad
	m{\cdot} a_k:=\begin{cases}
        1 &\text{if $m = k$}, \\
		k &\text{if $m = 1$}, \\
		m &\text{if $m\ne 1,k$},
	\end{cases}\quad\text{for $2\le k\le n-1$}.
\]
Figure~\ref{fig:flower} shows the automaton $\mF_n$ (and explains its name).
\begin{figure}[ht]
  \begin{center}
  \begin{tikzpicture}[->,>=latex,shorten >=1pt,node distance=2.8cm,auto,every state/.style={circle, draw=blue, minimum size=8mm, inner sep=0pt}]
    \node[state] (1) {$1$};
    \node[state] (2) [above right=of 1] {$n{-}3$};
    \node[state] (4) [below right=of 1] {$n{-}1$};
    \node[state] (1a) [below left=of 1] {$2$};
    \node[state] (1b) [left=of 1] {$3$};
    \node[state] (1c) [above left=of 1] {$4$};
    \node[state] (1r) [right=of 1] {$n{-}2$};
    \node[state] (0) [below=of 1] {$0$};
    \node (dots) [above =of 1] {$\mathbf{\centerdot\,\centerdot\,\centerdot}$};

    \draw (1) edge[pos=0.7] node[right] {$a_1$} (0)
    (1) edge [bend left=15,sloped,pos=0.85] node {$a_{n-3}$} (2)
    (2) edge [bend left=15,sloped,pos=0.4] node {$a_{n-3}$} (1)
    (1) edge [bend left=15,sloped,pos=0.85] node {$a_2$} (1a)
    (1a) edge [bend left=15,sloped,pos=0.4] node {$a_2$} (1)
    (1) edge [bend left=15,sloped,pos=0.6] node {$a_3$} (1b)
    (1b) edge [bend left=15,sloped,pos=0.4] node {$a_3$} (1)
    (1) edge [bend left=15,sloped,pos=0.6] node {$a_4$} (1c)
    (1c) edge [bend left=15,sloped,pos=0.25] node {$a_4$} (1)
    (1) edge [bend left=15,sloped,pos=0.6] node {$a_{n-2}$} (1r)
    (1r) edge [bend left=15,sloped,pos=0.4] node {$a_{n-2}$} (1)
    (1) edge [bend left=15,sloped,pos=0.55] node {$a_{n-1}$} (4)
    (4) edge [bend left=15,sloped,pos=0.15] node {$a_{n-1}$} (1);
 \end{tikzpicture}
\end{center}
  \caption{The DFA $\mF_n$. When at a state $q$, no outgoing edge labeled by a letter $x$ is shown, then a loop labeled $x$ at $q$ is assumed; such loops are omitted to lighten the figure.}
  \label{fig:flower}
\end{figure}

Alice wins the $2/\omega$-game on the automaton $\mF_n$ by using the following strategy. After Bob's move, as long as there remains a token on a state $k\ne0$, she plays the letter $a_1$ if $k=1$ and the word $a_ka_1$ if $k>1$. In this way, Alice eventually brings all tokens to state 0, regardless of Bob's responses.

Clearly, the word $a_1a_2a_1a_3\cdots a_{n-1}a_1$ with $a_1$ in all odd positions and $a_k$ in position $2(k-1)$, $2\le k\le n-1$, has length $2n-3$ and resets $\mF_n$.
Now suppose that $w$ is a reset word of minimum length for $\mF_n$. For any nonempty subset $S$ of $\mathbb{Z}_n$, we have
\[
|S{\cdot}a_1|=\begin{cases}
		|S|-1 &\text{if $0,1\in S$}, \\
		|S| &\text{otherwise},
	\end{cases}\quad \text{and} \quad |S{\cdot}a_k|=|S|\ \text{ for \ $2\le k\le n-1$}.
\]
Therefore, to compress the $n$-element set $\mathbb{Z}_n$ to a singleton, the word $w$ must have at least $n-1$ occurrences of the letter $a_1$. Each two of these occurrences must be separated by a nonempty word. Indeed, if
$w$ contained a factor $a_1^i$ with $i>1$, then replacing this factor by a single occurrence of $a_1$ would produce a shorter word that acts in exactly the same way as $w$, contradicting the minimality of $w$. Thus, $w$ can be written in the form $w=w_0a_1w_1a_1\cdots w_{n-2}a_1w_{n-1}$ for some nonempty words $w_1,\dots,w_{n-2}$. The total length of $w_1,\dots,w_{n-2}$ is at least $n-2$, whence $ |w| \ge  n-1 + n-2 = 2n - 3$.

Thus, the \rl{} of the automaton $\mF_n$ is $2n-3$.
\end{remark}

We do not know whether the bound $m(n-2)+1$ given by Proposition~\ref{prop:iteration} is tight for $m\ge 3$, and we tend to believe that it is not.

Proposition~\ref{prop:iteration} also leads to a new class of \sa{} satisfying the \v{C}ern\'y conjecture bound:

\begin{corollary}
\label{cor:iteration}
The \rl{} of an $n$-state automaton on which Alice has a winning strategy in the $n/\omega$-game does not exceed $n(n-2)+1=(n-1)^2$.
\end{corollary}

Although the observation in Corollary~\ref{cor:iteration} appears to be worth recording, it is fair to say that, in contrast to Theorem~\ref{thm:main}, we have so far not found any application of this observation to DFAs not defined in terms of synchronization games.

\subsection{Reset thresholds of $A_k$-automata with $k\in\mathbb{N}$}
\label{subsec:kspeed}
Here, for each $k\in\mathbb{N}$, we present a series of examples showing that, in contrast to $A_\omega$-automata, $A_k$-automata with $k\in\mathbb{N}$ can have reset thresholds vastly exceeding the number of states. More precisely,  the \rl{} of the DFAs in our series is bounded from below by a quadratic function of the number of states, with leading coefficient $\frac{1}{2k}$.

Since the \v{C}ern\'{y} series $\{\mathrsfs{C}_{n}\}_{n=2,3,\dots}$ witnesses the claim for $k=1$, we assume that $k>1$. For each $m\in\mathbb{N}$, consider the automaton $\mL_m^k:=(\mathbb{Z}_n,\,\{a_1,\dots,a_{k+m}\})$ with $n=k+m+1$, in which each letter $a_i$ with $i=1,\dots,k$ sends state $i$ to 0 and fixes all other states, whereas each letter $a_{k+j}$ with $j=1,\dots,m$ cyclically permutes the states $j$, $j+1$, \dots, $k+j-1$, $k+j$ and fixes all other states. Thus, for each $q\in\mathbb{Z}_n$, we have
\[
q{\cdot}a_i =\begin{cases}0 & \text{if } q=i,\\
q & \text{otherwise},
\end{cases}
\text{$i=1,\dots,k$, \ and} \ \
q{\cdot} a_{k+j} =
\begin{cases}
q+1 & \text{if } j \le q < k+j,\\
j   & \text{if } q = k+j,\\
q   & \text{otherwise},
\end{cases}
\text{$j=1,\dots,m$.}
\]
Notice that, in particular,
\begin{equation}\label{eq:decrease}
q{\cdot}a_q =\begin{cases}0 & \text{if } 1\le q\le k,\\
q-k & \text{if } k+1\le q\le k+m,
\end{cases}\quad\text{and } \ p{\cdot}a_q\ge p \ \text{ whenever } \ p\ne q.
\end{equation}
Notice also that 0 is a \emph{sink} of $\mL_m^k$, that is, a state fixed by all letters.

For illustration, Figure~\ref{fig:lipinseries} shows the automaton $\mL_2^k$.
\begin{figure}[htb]
\begin{center}
\begin{tikzpicture}[  >=latex,shorten >=1pt,
  node distance=2.3cm,auto,
  every state/.style={circle, draw=blue, scale=1},
  ell/.style={inner sep=0pt}]

\node[state] (1) {$1$};
\node[state] (2) [right of=1] {$2$};
\node[ell] (d1) [right of=2] {$\cdots$};
\node[state] (k) [right of=d1] {$k$};
\node[state, right of=k] (kp1) {$k{+}1$};
\node[state, right of=kp1] (kp2) {$k{+}2$};
\node[state] (0) [below of=d1] {$0$};
\node[ell] (d2) [below of=d1,yshift=1.25cm] {$\cdots$};

\draw[->] (1) to[bend right=20] node[left] {$a_1$} (0);
\draw[->] (2) to[bend right=10] node[left] {$a_2$} (0);
\draw[->] (k) to[bend left=20] node[right] {$a_k$} (0);
\draw[->] (1) to node[above] {$a_{k+1}$} (2);
\draw[->] (2) to node[above] {$a_{k+1}$} node[below] {$a_{k+2}$} (d1);
\draw[->] (d1) to node[above] {$a_{k+1}$} node[below] {$a_{k+2}$} (k);
\draw[->] (k) to node[above] {$a_{k+1}$} node[below] {$a_{k+2}$}(kp1);
\draw[->] (kp1) to node[below] {$a_{k+2}$} (kp2);
\draw[->] (kp1) to[bend right=30] node[sloped, above, pos=.6] {$a_{k+1}$} (1);
\draw[->] (kp2) to[bend right=30] node[sloped, above, pos=.4] {$a_{k+2}$} (2);
\end{tikzpicture}
\end{center}
\caption{The DFA $\mL_2^k$. When at a state $q$, no outgoing edge labeled by a letter $x$ is shown, then a loop labeled $x$ at $q$ is assumed; such loops are omitted to lighten the figure.}
\label{fig:lipinseries}
\end{figure}

\begin{proposition}\label{prop:PRTk1}
The DFA $\mL_m^k$ is an $A_k$-automaton.
\end{proposition}

\begin{proof}
Assume that Bob moves first. Denote by $w_s$ the history up to Bob's $s$-th turn, and let $A_s:=Q{\cdot}w_s$. If $A_s=\{0\}$, Alice has already won. Otherwise, if $A_s\ne\{0\}$ and $q_s:=\min(A_s\setminus\{0\})$, Alice plays the letter $a_{q_s}$. Then $q_s{\cdot}a_{q_s} = 0$ or $q_s{\cdot}a_{q_s} = q_s - k$ by \eqref{eq:decrease}. Whichever word $v_{s+1}$ of length $<k$ Bob plays in response, we have $q_s{\cdot}a_{q_s}v_{s+1}<q_s$. Hence, either $q_s{\cdot}a_{q_s}v_{s+1} = 0$ (and then $|A_{s+1}| < |A_s|$) or $q_{s+1}<q_s$. Therefore, the strategy guarantees that eventually $A_t = \{0\}$ for some $t$, and Alice wins.
\end{proof}

\begin{proposition}\label{prop:quadratic}
The reset threshold of the automaton\/ $\mL_m^k$ \ is not less than \ $\frac{n(n-1)}{2k}$.
\end{proposition}
\begin{proof}
Suppose $w = x_1\cdots x_\ell$ is a reset word for $\mL_m^k$. For each $s\in\{0,1,\dots,\ell\}$, denote by $u_s$ the prefix of $w$ of length $s$; in particular, $u_0=\varepsilon$.

The reset word $w$ brings all states of $\mL_m^k$ to a single state, and since 0 is fixed by all input letters, this state must be 0. Thus, $q{\cdot}w = 0$ for every state $q\in\mathbb{Z}_n$. Let $p_t:= q{\cdot}u_t$ and consider the sequence of states
\[
q=p_0,\ p_1,\ \dots \ p_{t-1},\ p_t, \dots\ p_{\ell}=q{\cdot}w=0.
\]
If $q\ne 0$, then some states in this sequence must be strictly greater than their immediate successors. Thus, for every $q\in\{1,\dots,k+m\}$, the set
\[
E_q:=\{s\in\{1,\dots,\ell\} : p_s < p_{s-1}\}
\]
is nonempty. For $s\in E_q$, we have $p_{s-1}\ne 0$. Since $p_s=p_{s-1}{\cdot}x_s$, the first part of \eqref{eq:decrease} implies that $p_{s-1}-p_s=k$ if $k+1\le p_{s-1}\le k+m$ and $p_{s-1}-p_s=p_{s-1}\le k$ if $1\le p_{s-1}\le k$. Thus, $0<p_{s-1}-p_s\le k$ for every $s\in E_q$, while $p_{t-1}-p_t\le 0$ for all $t\notin E_q$ by the inequality in~\eqref{eq:decrease}. Now we have
\begin{align*}
q&=\sum_{t=1}^\ell(p_{t-1}-p_t) &&\text{(telescopic sum)}\\
 &\le\sum_{s\in E_q}(p_{s-1}-p_s) &&\text{(dropping nonpositive summands)}\\
 &\le\sum_{s\in E_q}k=k|E_q| &&\text{(using that $p_{s-1}-p_s\le k$ for $s\in E_q$).}
\end{align*}
Hence $|E_q|\ge\frac{q}{k}$.

For any distinct $q,q'\in\{1,\dots,k+m\}$, the sets $E_q$ and $E_{q'}$ are disjoint. Indeed, suppose to the contrary that $E_q\cap E_{q'}\ne\varnothing$, and let $s\in E_q\cap E_{q'}$. Then $p_s=p_{s-1}{\cdot}x_s < p_{s-1}$ and, at the same time, $p'_s=p'_{s-1}{\cdot}x_s < p'_{s-1}$, where $p_t:= q{\cdot}u_t$ and $p'_t:= q'{\cdot}u_t$. By \eqref{eq:decrease}, for each letter $x\in\{a_1,\dots,a_{k+m}\}$, there is exactly one state $p\in\mathbb{Z}_n$ such that $p{\cdot}x < p$. Hence, $p_{s-1}= p'_{s-1}$, that is, $q{\cdot}u_{s-1}= q'{\cdot}u_{s-1}$. By the construction of the DFA $\mL_m^k$, if two distinct states are mapped to the same state by a word, then the resulting state is 0. Since $p_{s-1}\ne0$, this is a contradiction.

We thus have
\[
|w| \ge \left|\bigcup_{q=1}^{k+m}E_q\right|=\sum_{q=1}^{k+m}|E_q|\ge\sum_{q=1}^{k+m}\frac{q}{k}=\frac{(k+m+1)(k+m)}{2k}=\frac{n(n-1)}{2k}.
\]
Hence, $\rt(\mL_m^k)\ge\frac{n(n-1)}{2k}$, as claimed. 	
\end{proof}

\begin{remarks} 1. Question 2 in~\cite{NCMA} asked whether every A-automaton with $n$ states can be reset by a word of length $n-1$. (In our present terminology, the A-automata of~\cite{NCMA} are precisely DFAs on which Alice has a winning strategy in a variant of the 2-game where Bob is not allowed to skip moves, that is, to play the empty word.) Propositions~\ref{prop:PRTk1} and~\ref{prop:quadratic} imply a negative answer for $n>4$.

2. Since $|E_q|$ is an integer, the proof of Proposition~\ref{prop:quadratic} shows that in fact $|E_q|\ge\left\lceil\frac{q}{k}\right\rceil$ for each~$q$. Hence, $\rt(\mL_m^k)\ge L:=\sum\limits_{q=1}^{k+m} \left\lceil\frac{q}{k}\right\rceil$. We present a reset word of length $L$ for $\mL_m^k$, thus showing that $L$ is the precise value of $\rt(\mL_m^k)$.

Let $v_q:=a_q$ for $1\le q\le k$ and $v_q:= a_qa_{q-k}a_{q-2k}\cdots a_{t}$, where $1\le t\le k$, for $k+1\le q\le k+m$. Then $|v_q| = \left\lceil \frac{q}{k} \right\rceil$ and $q{\cdot}v_q=0$, whereas $r{\cdot}v_q=r$ for all $r>q$. Hence, the word $v_1v_2 \cdots v_{k+m}$ resets $\mL_m^k$  and has length $L$.
\end{remarks}

Combining the observation in the second remark with Proposition~\ref{prop:quadratic} yields
\[
\frac{n(n-1)}{2k} \le \rt(\mL_m^k) < \frac{n(n-1)}{2k} + n,
\]
since $\lceil\alpha\rceil<\alpha+1$ for every real number $\alpha$. Setting $e:=\left\lfloor\frac{n-1}{k}\right\rfloor$ and $r:=n-1-ke$, we can express $\rt(\mL_m^k)$ as follows:
\[
 \rt(\mL_m^k)=k\frac{e(e+1)}2+r(e+1).
\]
If $k$ divides $n-1$, that is, $r=0$ and $e=\frac{n-1}{k}$, the right-hand side simplifies to $\frac{(n-1+k)(n-1)}{2k}$.

If $\mathbf{S}$ is a class of DFAs, its \emph{\v{C}ern\'{y} function} $\mathfrak{C}_{\mathbf{S}}(n)$ is defined as the maximum \rl\ among all $n$-state \sa\ from $\mathbf{S}$. With this notation, Theorem~\ref{thm:main}, together with the remark following it, establishes the equality $\mathfrak{C}_{\mathbf{A}_\omega}(n)=n-1$, while Proposition~\ref{prop:quadratic} yields the inequality $\mathfrak{C}_{\mathbf{A}_k}(n)\ge\frac{n(n-1)}{2k}$. We do not know how close the lower bound $\frac{n(n-1)}{2k}$ is to the actual value of $\mathfrak{C}_{\mathbf{A}_k}(n)$.

We have assumed that $k>1$ to avoid overlaps with already known facts, but the construction of the automaton $\mL_m^k$ makes perfect sense also for $k=1$. For each $n\ge 3$, the DFA $\mL_{n-2}^1$ coincides with the automaton constructed by Igor Rystsov (see~\cite[Theorem 2]{Rystsov_rus:1977} or \cite[Theorem 6.1]{Rystsov:1997}) to establish the lower bound $\binom{n}{2}$ for the \v{C}ern\'{y} function of the class of all \sa{} with a sink; this bound is known to be tight. This gives some supporting evidence for the conjecture that $\frac{(n-1+k)(n-1)}{2k}$ may be the correct value of the \v{C}ern\'{y} function of the class of $A_k$-automata with a sink for $n-1$ divisible by $k$.

If $k=n-1$, the expression $\frac{(n-1+k)(n-1)}{2k}$ evaluates to $n-1$. We will show in Section~\ref{sec:nstates} that each $A_{n-1}$-automaton with a sink is an $A_{\omega}$-automaton, and so has \rl{} at most $n-1$ by Theorem~\ref{thm:main}. Thus, $\frac{(n-1+k)(n-1)}{2k}$ gives the correct value of the \v{C}ern\'{y} function of the class of $A_k$-automata with a sink for both extreme cases $k=1$ and $k=n-1$.

\section{Decidability}
\label{sec:decidability}

Here we address the natural question of how to efficiently determine whether Alice has a winning strategy in the $k$-game on a given DFA. We assume that the reader is familiar with basics of algorithm theory, including Big O notation and the concepts of a polynomial-time algorithm and polynomial-time reduction.

\subsection{The 2-subset automaton} Our analysis relies on the well-known construction of the 2-subset automaton, which we now recall. Given a DFA $\mA=(Q,\Sigma)$, let $Q^{[2]}$ denote the set of all 2-element subsets of $Q$, augmented with a fresh symbol $\mathbf{s}$. Define the action $\circ$ of the letters $a\in\Sigma$ on $Q^{[2]}$ by:
 \[
\mathbf{s}{\circ}a:=\mathbf{s}\  \text{ and } \ \{q,q'\}{\circ}a:=\begin{cases}
 \{q{\cdot}a,\,q'{\cdot}a\}&\text{if } q{\cdot}a\ne q'{\cdot}a,\\
 \mathbf{s} &\text{otherwise}.
 \end{cases}
 \]
This defines a DFA $(Q^{[2]},\Sigma)$, in which $\mathbf{s}$ is a sink. We denote this automaton by $\mA^{[2]}$. If $\mA$ has $n$ states, then $\mA^{[2]}$ has $\binom{n}{2}+1$ states.

It is known (see, e.g., \cite[Theorem 3]{Rystsov:2003}, and is easy to verify) that the automata $\mA$ and $\mA^{[2]}$ have the same set of reset words. (In particular, either both automata are synchronizing or neither is.)
We can therefore apply to $\mA^{[2]}$ the following observation.

\begin{lemma}
\label{lem:equivalence}
For any $k\in\overline{\mathbb{N}}$ and any \sa{} $\mA$ and $\mA'$ with the same sets of reset words, either both $\mA$ and $\mA'$ are $A_k$-automata or neither is.
\end{lemma}

\begin{proof}
Let $\mathfrak{S}$ be Alice's winning strategy in the $k$-game on $\mA$. Alice can mimic $\mathfrak{S}$ in the $k$-game on $\mA'$ as follows: whenever Bob makes a move in the game on $\mA'$, she reproduces the same move on $\mA$, determines which letter $\mathfrak{S}$ prescribes as her response, and plays the same letter in the game on $\mA'$. Then, in the games on $\mA$ and $\mA'$, the histories at each turn coincide. If the final history of the game on $\mA$ is a reset word for $\mA$, it also is such for $\mA'$. Hence $\mA'$ is an $A_k$-automaton whenever $\mA$ is, and converse follows by symmetry.
\end{proof}

As mentioned, Lemma~\ref{lem:equivalence} applies to $\mA$ and $\mA^{[2]}$, yielding the following useful fact.

\begin{lemma}
\label{lem:2subsets}
For any $k\in\overline{\mathbb{N}}$ and any DFA $\mA$, either both $\mA$ and $\mA^{[2]}$ are $A_k$-automata or neither is.
\end{lemma}

For any $k\in\overline{\mathbb{N}}$, Lemma~\ref{lem:2subsets} provides a polynomial-time reduction from the problem of recognizing $A_k$-automata to the problem of recognizing $A_k$-automata with a sink. Therefore, any polynomial-time algorithm for the latter immediately gives a polynomial-time algorithm for the former. Since automata with a sink are technically simpler to handle, we approach the recognition problem for $A_k$-automata via its restriction to DFAs with a sink, both for $k=\omega$ and for $k\in\mathbb{N}$.

\subsection{Polynomial decidability of $\mathbf{A}_\omega$} We begin with a characterization of $A_\omega$-automata with a sink. Its statement involves the concept of a \scc, which we now recall. Two states $q$ and $q'$ of a DFA $\mA=(Q,\Sigma)$ are said to be \emph{mutually reachable} if $q'=q{\cdot}w$ and $q=q{'\cdot}w'$ for some words $w,w'\in\Sigma^*$. The mutual reachability relation is an equivalence on $Q$, and its classes are called the \emph{strongly connected components} of $\mA$.
\begin{proposition}
\label{prop:criterion}
A DFA $\mA=(Q,\Sigma)$ with a sink $\mathbf{s}$ is an $A_\omega$-automaton if and only if for every state $q\ne\mathbf{s}$ there exists a letter $a\in\Sigma$ such that the states $q$ and $q{\cdot}a$ belong to distinct \scc{} of $\mA$.
\end{proposition}

\begin{proof}
We argue using the graphical representation of synchronization games, in which all states initially hold tokens. The token on the state $\mathbf{s}$ never moves in the course of the game.

Suppose that there exists a state $q\ne\mathbf{s}$ such that for every $a\in\Sigma$ the state $q{\cdot}a$ lies in the \scc{} of $q$. Then, for each $a\in\Sigma$, there exists a word $w_a\in\Sigma^*$ such that $q{\cdot}aw_a=q$. Therefore, in the $\omega$-game on $\mA$, Bob can always return the token to $q$ by responding with $w_a$ whenever Alice plays $a$. Hence, after each of Bob's moves, at least the states $\mathbf{s}$ and $q$ hold tokens, and Bob wins.

Conversely, suppose that for every state $q\ne\mathbf{s}$ there exists a letter $a\in\Sigma$ such that  the states $q$ and $q{\cdot}a$ belong to distinct \scc{} of $\mA$. Then Alice wins the $\omega$-game on $\mA$ from any position in which some states other than $\mathbf{s}$ still hold tokens, by repeatedly using the following procedure. She chooses an arbitrary state $q_0\ne\mathbf{s}$ holding a token and plays a letter $a_1$ such that $q_0{\cdot}a_1\notin K_0$ where $K_0$ is the \scc{} containing $q_0$. When Bob plays a word $w_1$ and $q_0{\cdot}a_1w_1\ne\mathbf{s}$, Alice responds with a letter $a_2$ such that $q{\cdot}a_1w_1a_2\notin K_1$ where $K_1$ is the \scc{} containing $q_0{\cdot}a_1w_1$, and so on. Since the \scc{}s in the sequence $K_0,K_1,\dots$ can never repeat, this process must eventually terminate, which is only possible if $q_0{\cdot}a_1w_1a_2\cdots a_\ell w_\ell=\mathbf{s}$ for some $\ell$. The token that traveled from $q_0$ to $\mathbf{s}$ is then removed. Acting in this way, Alice removes all tokens except the one on $\mathbf{s}$ and wins.
\end{proof}

\begin{proposition}
\label{prop:omega-algorithm-sink}
There exists an algorithm that, given a DFA $\mathrsfs{A}=(Q,\Sigma)$ with a sink, decides the winner in the $\omega$-game on $\mathrsfs{A}$ in $O(|Q|\cdot|\Sigma|)$ time.
\end{proposition}

\begin{proof}
Let $n:=|Q|$ and $m:=|\Sigma|$. We show that the condition of Proposition~\ref{prop:criterion} can be verified in $O(nm)$ time. It is well known that the \scc{}s of a digraph can be computed in time linear in the numbers of vertices and edges (for example, by Tarjan's algorithm~\cite[Section 4]{Tarjan:1972}). The digraph of $\mA$ has $n$ vertices and $nm$ edges so its \scc{}s can be computed in $O(nm)$ time.

If the computation reveals that $\mA$ has more than one sink, then $\mA$ cannot be a \san, and Bob wins. If $\mA$ has exactly one sink, $O(nm)$ checks suffice to verify, for each non-sink state $q$ and each letter $a\in\Sigma$, whether the states $q$ and $q{\cdot}a$ belong to the same \scc.
\end{proof}

\begin{theorem}
\label{thm:omega-algorithm}
There exists an algorithm that, given a DFA $\mathrsfs{A}=(Q,\Sigma)$, decides the winner in the $\omega$-game on $\mathrsfs{A}$ in $O(|Q|^2\cdot|\Sigma|)$ time.
\end{theorem}

\begin{proof}
By Lemma~\ref{lem:2subsets} (with $k=\omega$) it suffices to recognize whether the 2-subset automaton $\mA^{[2]}$ is an $A_\omega$-automaton. Since $\mA^{[2]}$ has a sink, Proposition~\ref{prop:omega-algorithm-sink} applies, providing an answer in $O(|Q^{[2]}|\cdot|\Sigma|)= O(|Q|^2\cdot|\Sigma|)$ time. The preprocessing overhead (namely, constructing $\mA^{[2]}$ from $\mA$) is also $O(|Q|^2\cdot|\Sigma|)$, so the total running time is $O(|Q|^2\cdot|\Sigma|)$.
\end{proof}

\subsection{Polynomial decidability of $\mathbf{A}_k$ with $k\in\mathbb{N}$}
As in the previous subsection, we first provide an algorithm for DFAs with a sink. Unlike the algorithm of Proposition~\ref{prop:omega-algorithm-sink}, which is simple enough to admit an informal description, the present algorithm is more involved and therefore requires a more formal presentation. We begin with pseudocode and then give a justification and a complexity analysis.

We introduce notation used in the following pseudocode. For a nonempty subset $F\subseteq Q$ and a state $q\in Q$ of a DFA $\mA=(Q,\Sigma)$, define
\[
F{\cdot}\Sigma^{-1}:=\bigcup_{a\in\Sigma}\{q\in Q:q{\cdot}a\in F\}\quad \text{and} \quad q{\cdot}\Sigma^{<k}:=\{q{\cdot}w:w\in\Sigma^{<k}\}.
\]
In terms of the graphical representation of the $k$-game on $\mA$, the set $F{\cdot}\Sigma^{-1}$ consists of all states from which Alice can move a token to a state in $F$ in a single move, whereas $q{\cdot}\Sigma^{<k}$ collects all states to which Bob can move the token from $q$ in a single move.
\begin{algorithm}
  \setlength{\commentspace}{6cm}
  \begin{algorithmic}[1]
    \FUNC{Input: {\normalfont DFA $\mA=(Q,\Sigma)$ with a sink $\mathbf{s}$}}
    \STATE\algcomment{0}{Initialization}$F,P\gets\{\mathbf{s}\}$
    \WHILE{$F\subsetneqq Q$}
    \STATE\algcomment{1}{Preliminary marking}$P\gets F{\cdot}\Sigma^{-1}$
    \IF{$(\forall q\in P\setminus F)\ q{\cdot}\Sigma^{<k}\setminus P\ne\varnothing$}
    \RETURN $\mA$ is not an $A_k$-automaton
    \ELSE\STATE\algcomment{2}{Firm marking}$F\gets F\cup\{q\in P\setminus F: q{\cdot}\Sigma^{<k}\subseteq P\}$
    \ENDIF
    \ENDWHILE
    \RETURN $\mA$ is an $A_k$-automaton
  \end{algorithmic}
\caption{Checking whether Alice has a winning strategy in the $k$-game on a given DFA with a sink.}\label{alg:checkforAk}
\end{algorithm}

The reader may find it instructive to trace Algorithm~\ref{alg:checkforAk} on the automaton $\mL^3_2$ from the series introduced in Subsection~\ref{subsec:kspeed}. Figures~\ref{fig:algorithm} and~\ref{fig:algorithm-end} illustrate the execution of the algorithm under the assumption that $k=3$.

Algorithm~\ref{alg:checkforAk} operates with two sets of states: the set $F$ of firmly marked states and the set $P$ of preliminarily marked states. Lines 3 and 7 show how these sets are updated. (In Figures~\ref{fig:algorithm} and~\ref{fig:algorithm-end} the states in $F$ are shown in dark blue, whereas states in $P$ are shown in light blue, with the darker color overriding the lighter.) The algorithm initializes $F$ as the singleton~$\{\mathbf{s}\}$. In the main loop, the set $P$ is first updated to consist of all tails of the edges in $\mA$ whose heads lie in the current set $F$. Then each state in $P\setminus F$ is tested for the property that all its images under the action of an arbitrary word of length less than $k$ belong to $P$. If no such states exist, the algorithm reports that $\mA$ is not an $A_k$-automaton and stops; otherwise, all states that pass the test are added to the set $F$. The main loop repeats until $F=Q$, at which point the algorithm reports that $\mA$ is an $A_k$-automaton.

In the top picture in Figure~\ref{fig:algorithm}, state 0 is marked dark blue, since 0 is a sink of the DFA $\mL^3_2$. The next picture shows states 1, 2, and 3 marked light blue due to the edges $1\xrightarrow{a_1}0$,
$2\xrightarrow{a_2}0$, and $3\xrightarrow{a_3}0$, respectively. Next, the test is performed, and only state 1 passes, whereas states 2 and 3 do not, since $2{\cdot}a_4^2=3{\cdot}a_4=4\notin\{0,1,2,3\}$. At the next step, state~4 is marked light blue due to the edge $4\xrightarrow{a_4}1$. Next, states 2, 3, and 4 are tested, and this time state 2 passes---the equality $2{\cdot}a_4^2=4$ is no longer an obstacle, since state~4 has been added to the set of preliminarily marked states. States 3 and 4 do not pass, since $3{\cdot}a_5^2=4{\cdot}a_5=5\notin\{0,1,2,3,4\}$. This outcome is shown in the bottom picture in Figure~\ref{fig:algorithm}.
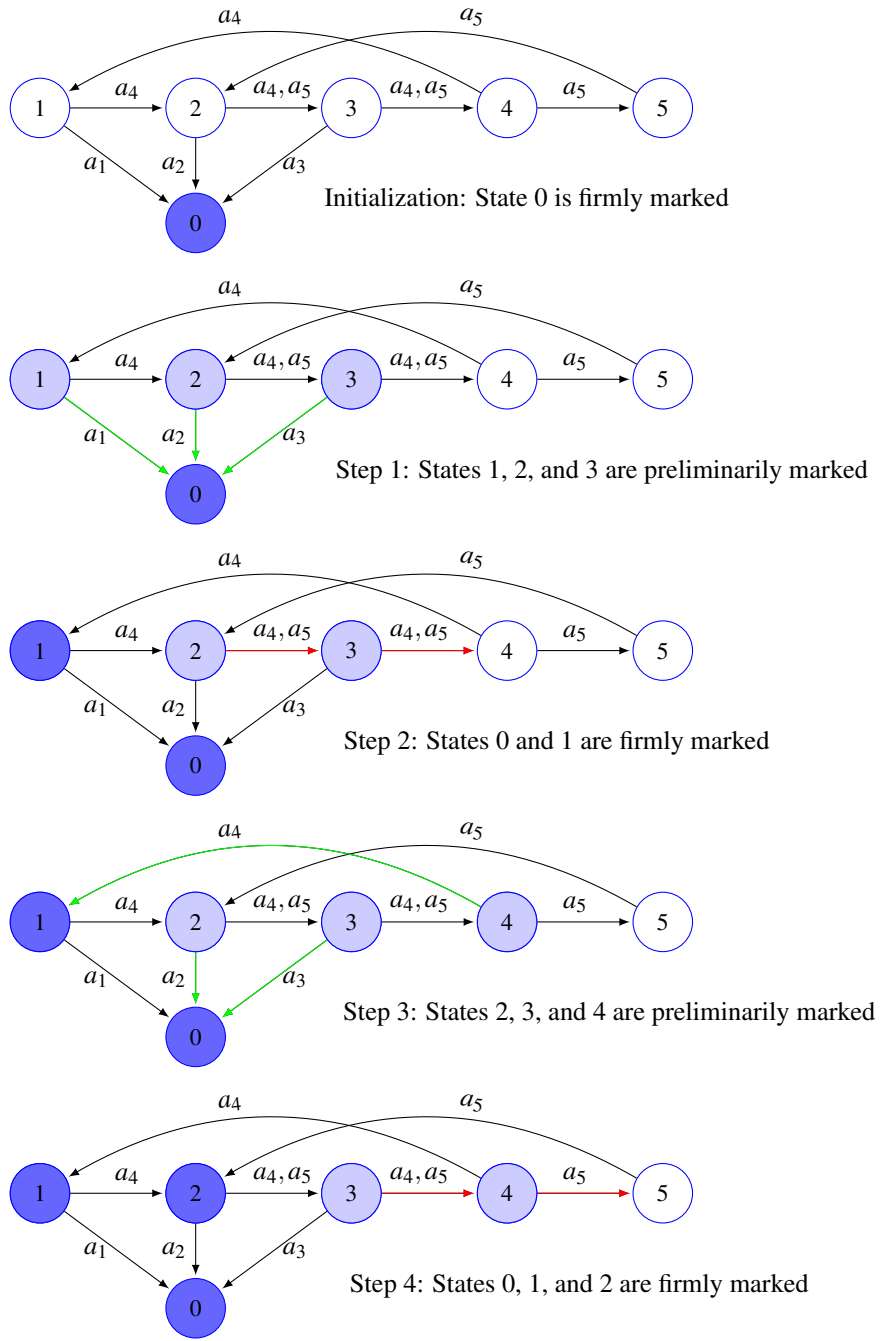
\begin{figure}[p]
\begin{center}
\begin{tikzpicture}[
  >=latex, shorten >=1pt,
  node distance=2.3cm, auto,
  every state/.style={circle, draw=blue, scale=0.9},
  ell/.style={inner sep=0pt}
]

\def\dy{3.61cm}

\foreach \i in {0,1,2,3,4} {
  \begin{scope}[yshift=-\i*\dy]

    \node[state] (1-\i) {$1$};
    \node[state] (2-\i) [right of=1-\i] {$2$};
    \node[state] (3-\i) [right of=2-\i] {$3$};
    \node[state] (4-\i) [right of=3-\i] {$4$};
    \node[state] (5-\i) [right of=4-\i] {$5$};
    \node[state, fill=blue!60] (0-\i) [below of=2-\i,yshift=0.6cm] {$0$};

    \draw[->] (1-\i) to node[left] {$a_1$} (0-\i);
    \draw[->] (2-\i) to node[left] {$a_2$} (0-\i);
    \draw[->] (3-\i) to node[right] {$a_3$} (0-\i);

    \draw[->] (1-\i) to node[above, pos=0.6] {$a_4$} (2-\i);
    \draw[->] (2-\i) to node[above, pos=0.6] {$a_4,a_5$} (3-\i);
    \draw[->] (3-\i) to node[above, pos=0.4] {$a_4,a_5$} (4-\i);
    \draw[->] (4-\i) to node[above, pos=0.4] {$a_5$} (5-\i);

    \draw[->] (4-\i) to[bend right=30]
      node[sloped, above, pos=.6] {$a_4$} (1-\i);
    \draw[->] (5-\i) to[bend right=30]
      node[sloped, above, pos=.4] {$a_5$} (2-\i);

  \end{scope}
}

\draw[->,green] (1-1) to (0-1);
\draw[->,green] (2-1) to (0-1);
\draw[->,green] (3-1) to (0-1);

\draw[->,red] (2-2) to (3-2);
\draw[->,red] (3-2) to (4-2);

\draw[->,green] (2-3) to (0-3);
\draw[->,green] (3-3) to (0-3);
\draw[->,green] (4-3) to[bend right=30] (1-3);

\draw[->,red] (4-4) to (5-4);
\draw[->,red] (3-4) to (4-4);

\node (s0) [right of=0-0,xshift=2.1cm,yshift=0.3cm] {Initialization: State 0 is firmly marked};
\node (s1) [right of=0-1,xshift=3.1cm,yshift=0.3cm] {Step 1: States 1, 2, and 3 are preliminarily marked};
\node (s2) [right of=0-2,xshift=2.5cm,yshift=0.3cm] {Step 2: States 0 and 1 are firmly marked};
\node (s3) [right of=0-3,xshift=3.2cm,yshift=0.3cm] {Step 3: States 2, 3, and 4 are preliminarily marked};
\node (s4) [right of=0-4,xshift=2.8cm,yshift=0.3cm] {Step 4: States 0, 1, and 2 are firmly marked};

\node[state, fill=blue!20] (s11)  [left of=2-1] {$1$};
\node[state, fill=blue!20] (s12)  [right of=s11] {$2$};
\node[state, fill=blue!20] (s13)  [right of=s12] {$3$};

\node[state, fill=blue!60] (s21)  [left of=2-2] {$1$};
\node[state, fill=blue!20] (s22)  [right of=s21] {$2$};
\node[state, fill=blue!20] (s23)  [right of=s22] {$3$};

\node[state, fill=blue!60] (s31)  [left of=2-3] {$1$};
\node[state, fill=blue!20] (s32)  [right of=s31] {$2$};
\node[state, fill=blue!20] (s33)  [right of=s32] {$3$};
\node[state, fill=blue!20] (s34)  [right of=s33] {$4$};

\node[state, fill=blue!60] (s41)  [left of=2-4] {$1$};
\node[state, fill=blue!60] (s42)  [right of=s41] {$2$};
\node[state, fill=blue!20] (s43)  [right of=s42] {$3$};
\node[state, fill=blue!20] (s44)  [right of=s43] {$4$};

\end{tikzpicture}
\end{center}
\caption{The first steps of Algorithm~\ref{alg:checkforAk} (with $k=3$) on the DFA $\mL_2^3$, with loops omitted. States in $F$ (firmly marked) are shown in dark blue, whereas states in $P\setminus F$ (preliminarily but not yet firmly marked) are shown in light blue. Green edges witness preliminary markings and red paths originate from the states that fail the test for firm marking.}
\label{fig:algorithm}
\end{figure}

The two final steps are illustrated in Figure~\ref{fig:algorithm-end}.

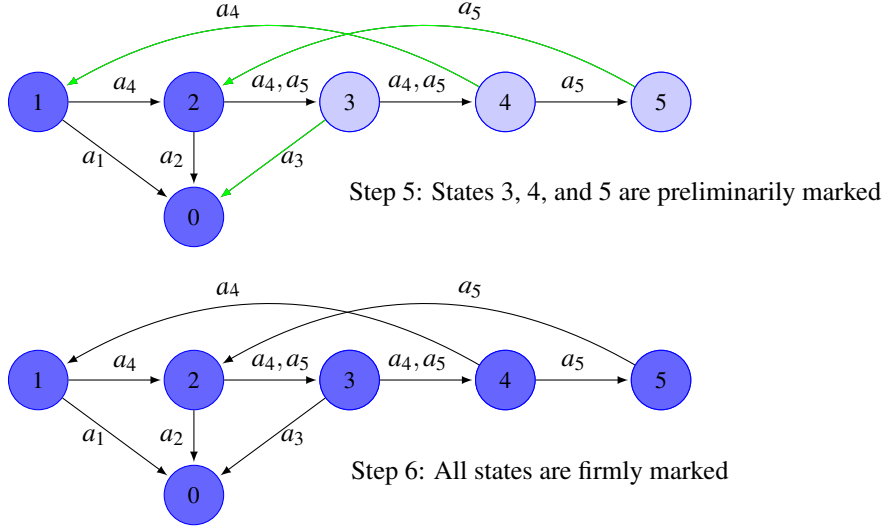
\begin{figure}[ht]
\begin{center}
\begin{tikzpicture}[
  >=latex, shorten >=1pt,
  node distance=2.3cm, auto,
  every state/.style={circle, draw=blue, scale=0.9},
  ell/.style={inner sep=0pt}
]

\def\dy{3.7cm}

\foreach \i in {0,1} {
  \begin{scope}[yshift=-\i*\dy]

    \node[state, fill=blue!60] (1-\i) {$1$};
    \node[state, fill=blue!60] (2-\i) [right of=1-\i] {$2$};
    \node[state] (3-\i) [right of=2-\i] {$3$};
    \node[state] (4-\i) [right of=3-\i] {$4$};
    \node[state] (5-\i) [right of=4-\i] {$5$};
    \node[state, fill=blue!60] (0-\i) [below of=2-\i,yshift=0.6cm] {$0$};

    \draw[->] (1-\i) to node[left] {$a_1$} (0-\i);
    \draw[->] (2-\i) to node[left] {$a_2$} (0-\i);
    \draw[->] (3-\i) to node[right] {$a_3$} (0-\i);

    \draw[->] (1-\i) to node[above, pos=0.6] {$a_4$} (2-\i);
    \draw[->] (2-\i) to node[above, pos=0.6] {$a_4,a_5$} (3-\i);
    \draw[->] (3-\i) to node[above, pos=0.4] {$a_4,a_5$} (4-\i);
    \draw[->] (4-\i) to node[above, pos=0.4] {$a_5$} (5-\i);

    \draw[->] (4-\i) to[bend right=30]
      node[sloped, above, pos=.6] {$a_4$} (1-\i);
    \draw[->] (5-\i) to[bend right=30]
      node[sloped, above, pos=.4] {$a_5$} (2-\i);

  \end{scope}
}

\draw[->,green] (3-0) to (0-0);
\draw[->,green] (4-0) to[bend right=30] (1-0);
\draw[->,green] (5-0) to[bend right=30] (2-0);

\node (s5) [right of=0-0,xshift=3.3cm,yshift=0.3cm] {Step 5: States 3, 4, and 5 are preliminarily marked};
\node (s6) [right of=0-1,xshift=2.3cm,yshift=0.3cm] {Step 6: All states are firmly marked};

\node[state, fill=blue!20] (s53)  [right of=2-0] {$3$};
\node[state, fill=blue!20] (s54)  [right of=s53] {$4$};
\node[state, fill=blue!20] (s55)  [right of=s54] {$5$};

\node[state, fill=blue!60] (s63)  [right of=2-1] {$3$};
\node[state, fill=blue!60] (s64)  [right of=s63] {$4$};
\node[state, fill=blue!60] (s65)  [right of=s64] {$5$};
\end{tikzpicture}
\end{center}
\caption{The final steps of Algorithm~\ref{alg:checkforAk} (with $k=3$) on the DFA $\mL_2^3$.}
\label{fig:algorithm-end}
\end{figure}

For comparison, when executed on the same DFA $\mL_2^3$ with $k=4$, Algorithm~\ref{alg:checkforAk} stops at Step~2 because each of the states 1, 2, and 3 can be sent to state 4 by a word of length less than 4, and hence none of these states passes the test.

\begin{proposition}
\label{prop:correctness}
Algorithm~\ref{alg:checkforAk} correctly determines the winner in the $k$-game on a given DFA with a sink.
\end{proposition}

\begin{proof}
We keep the notation of Algorithm~\ref{alg:checkforAk}. Let $F_i$ and $P_i$, $i=0,1,2,\dotsc$, denote the sets $F$ and $P$, respectively, after the $i$-th iteration of the algorithm when it is executed on the DFA $\mA=(Q,\Sigma)$. We prove by induction on $i$ that whenever a state $q\in F_i$ holds a token, then for every move $w\in\Sigma^{<k}$ played by Bob from $q$, Alice has a strategy to move the token to the sink~$\mathbf{s}$ from the state $q{\cdot}w$.

The claim trivially holds for $i=0$. Assume now that the claim holds for all states in $F_i$ and prove that it holds for each state $q\in F_{i+1}$. Clearly, we may assume that $q\in F_{i+1}\setminus F_i$. By the construction of the algorithm, we have $q{\cdot}\Sigma^{<k}\subseteq P_{i+1}$. Hence, if Bob plays a word $w\in\Sigma^{<k}$ from $q$, the state $q{\cdot}w$ belongs to $P_{i+1}$. By definition $P_{i+1}=F_i{\cdot}\Sigma^{-1}$, so there exists a letter $a\in\Sigma$ such that $(q{\cdot}w){\cdot}a\in F_i$. Alice plays this letter $a$, thereby moving the token into a state of $F_i$. By the inductive assumption, Alice can then force the token to reach the sink~$\mathbf{s}$, regardless of Bob's subsequent moves.

It follows that if the algorithm terminates with $F=Q$, then Alice has a winning strategy from every state. Combining these strategies, she can remove the tokens one by one from all states except the sink. Hence $\mA$ is an $A_k$-automaton.

Conversely, suppose that the algorithm stops at some $i$-th iteration without reaching the equality $F=Q$. Then $F_i\subsetneqq Q$. If $P_{i+1}=Q$, the condition $q{\cdot}w\in P_{i+1}$ holds for all states $q$ and words $w$, and the $(i+1)$-st iteration of the algorithm would proceed, contradicting termination. Hence, $P_{i+1}\subsetneqq Q$. Bob can maintain a token outside $P_{i+1}$ indefinitely by the following strategy. Place a token on a state $q\notin P_{i+1}$ and let Alice play an arbitrary letter $a\in\Sigma$. Then $q{\cdot}a\notin F_i$ by the definition of $P_{i+1}$. If $q{\cdot}a\notin P_{i+1}$, Bob skips his turn. If $q{\cdot}a\in P_{i+1}\setminus F_i$, then by the termination condition (line 4), there exists a word $w\in\Sigma^{<k}$ such that $(q{\cdot}a){\cdot}w\notin P_{i+1}$. Bob plays this word, thereby returning the token outside $P_{i+1}$. Hence a token is always present on some state different from~$\mathbf{s}$. Since Bob has a winning strategy, $\mA$ is not an $A_k$-automaton.
\end{proof}

For the complexity analysis of Algorithm~\ref{alg:checkforAk}, let $n:=|Q|$ and $m:=|\Sigma|$. The most time-consuming step of the algorithm is the test in line 4 of whether $q{\cdot}\Sigma^{<k}\subseteq P$ for all $q\in P\setminus F$. This test can be performed by using a breadth-first search (BFS) in the digraph of $\mA$, starting at $q$. In any graph $G=(V,E)$, BFS finds shortest-path distances from a given vertex in time $O(|V|+|E|)$; see \cite[Section 20.2]{Cormen&Leiserson&Rivest&Stein:2022}. Since the digraph of $\mA$ has $n$ vertices and $nm$ edges, a single BFS takes $O(nm)$ time. Once the BFS is completed, checking whether all states at distance less than $k$ from $q$ belong to $P$ can be done in~$O(n)$ time, assuming that the set $P$ is represented as a Boolean array indexed by states, via direct addressing; see \cite[Section 11.1]{Cormen&Leiserson&Rivest&Stein:2022}.

From this, one might conclude that the total time spent on the tests needed for firm marking on the entire execution of the algorithm is $O(n^2m)$ if  each non-sink state $q$ is tested against the condition $q{\cdot}\Sigma^{<k}\subseteq P$ only once. However, it is not immediately clear how to implement Algorithm~\ref{alg:checkforAk} as to guarantee this property. Observe that a naive implementation, as illustrated by the above example of running the algorithm on the DFA $\mL_2^3$, does not satisfy this requirement: for instance, state~3 is tested three times, failing the test on Steps~2 and 4 and passing it on Step 6.

To bypass this bottleneck, we introduce an auxiliary data structure that, for each non-sink state $q$, stores the set $R[q]:=q{\cdot}\Sigma^{<k}\setminus P$ consisting of all currently unmarked states reachable from $q$ by words of length less than $k$. Then the condition  $q{\cdot}\Sigma^{<k}\subseteq P$ is equivalent to the emptiness of $R[q]$.

The initial values of the sets $R[q]$ can be precomputed by BFS. Each iteration of the main loop then updates these sets by removing, from those that are still nonempty, the states the algorithm adds to $P$ during this iteration. The updates are made `on the fly', that is, as soon as a state is added to $P$, it is removed from each the sets $R[q]$ in which it occurs. A state $q\in P\setminus F$ is added to $F$ once the set $R[q]$ becomes empty.

The following presents the proposed modification of Algorithm~\ref{alg:checkforAk} in pseudocode.
\begin{algorithm}
  \setlength{\commentspace}{6cm}
  \begin{algorithmic}[1]
    \FUNC{Input: {\normalfont DFA $\mA=(Q,\Sigma)$ with a sink $\mathbf{s}$}}
    \STATE\algcomment{0}{Initialization}$F,P\gets\{\mathbf{s}\}$
    \FOR{each $q\in Q\setminus\{\mathbf{s}\}$}
        \STATE\algcomment{1}{Precomputation} $R[q]\gets q{\cdot}\Sigma^{<k}\setminus\{\mathbf{s}\}$
      \ENDFOR
    \WHILE{$F\subsetneqq Q$}
      \STATE\algcomment{1}{Preliminary marking}$P\gets F{\cdot}\Sigma^{-1}$
      \FOR{each $q\in P\setminus F$}
        \STATE $R[q]\gets R[q]\setminus P$
      \ENDFOR
      \IF{$(\forall q\in P\setminus F)\ R[q]\neq\varnothing$}
        \RETURN $\mA$ is not an $A_k$-automaton
      \ELSE
        \STATE\algcomment{2}{Firm marking}
        $F\gets F\cup\{q\in P\setminus F: R[q]=\varnothing\}$
      \ENDIF
    \ENDWHILE
    \RETURN $\mA$ is an $A_k$-automaton
  \end{algorithmic}
\caption{A modification of Algorithm~\ref{alg:checkforAk} with $O(|Q|^2\cdot|\Sigma|)$ time implementation}\label{alg:quadratic}
\end{algorithm}

For illustration, the following table shows the evolution of the sets $R[q]$, $P$, and $F$ during the execution of Algorithm~\ref{alg:quadratic} on the DFA $\mL_2^3$. The reader may refer again to Figures~\ref{fig:algorithm} and~\ref{fig:algorithm-end} for a visual support of this table.
\begin{center}
\begin{tabular}{c|c@{}c@{}c@{}c@{}c@{}c@{}c}
& Step 0 & Step 1 & Step 2 & Step 3 & Step 4 & Step 5 & Step 6\\
\hline
$R[1]$ & \{1,2,3\} && $\varnothing$ && &&\\
$R[2]$ & \{2,3,4\} && \{4\} && $\varnothing$  &&\\
$R[3]$ & \{3,4,5\} && \{4,5\} && \{5\} && $\varnothing$ \\
$R[4]$ & \{1,2,4,5\} && \{4,5\} && \{5\} && $\varnothing$ \\
$R[5]$ & \{2,3,5\} && \{4,5\} && \{5\} && $\varnothing$ \\
$P$ & \{0\} & 1,2,3 join $P$ && 4 joins $P$ && 5 joins $P$ &\\
$F$ & \{0\} && 1 joins $F$ && 2 joins $F$ && 3,4,5 join $F$
\end{tabular}
\end{center}

\begin{proposition}\label{prop:Ak-complexity}
Algorithm~\ref{alg:quadratic} runs on a DFA $Q,\Sigma)$ in time $O(|Q|^2\cdot|\Sigma|)$.
\end{proposition}

\begin{proof}
Let $n := |Q|$ and $m := |\Sigma|$.

For each state $q \in Q$, the algorithm initializes the set $R[q]$ by performing a BFS, starting at $q$. As discussed above, the BFS takes $O(nm)$ time. As this precomputation is carried out for all $q \in Q$, initializing all sets $R[q]$ requires $O(n^2m)$ time in total.

During the execution of the main loop, the algorithm updates the sets $R[q]$ only by removing states that are added to the set $P$. Each state $p \in Q$ is added to $P$ at most once, and whenever this happens, $p$ is immediately removed from every set $R[q]$ in which it occurs. Hence, each pair $(p,q)$ such that $p\in R[q]$ is processed at most once. Assuming that the sets $P$ and $R[q]$ are represented as Boolean arrays indexed by states, each addition or removal takes $O(1)$ time via direct addressing. Therefore, the cumulative cost of all updates to the sets $R[q]$ does not exceed the number of pairs $(p,q)$, and is thus $O(n^2)$.

In each iteration of the main loop, the algorithm checks whether $R[q]=\varnothing$ for all $q \in P\setminus F$. These checks take time linear in $|P \setminus F|$, and since each state enters $P\setminus F$ at most once, the total time spent on these tests is $O(n)$.

All remaining operations of the algorithm, including the preliminary marking step computing $F{\cdot}\Sigma^{-1}$, take $O(nm)$ time in total, as each transition is examined at most once.

Combining the above bounds, we conclude that the overall running time is dominated by the time spent on the precomputation of the sets $R[q]$, and is therefore $O(n^2m)$.
\end{proof}

The next result can be deduced from Lemma~\ref{lem:2subsets} (with $k\in\mathbb{N}$) in exactly the same way as Theorem~\ref{thm:omega-algorithm}, except that Propositions~\ref{prop:correctness} and~\ref{alg:quadratic} are used in place of Proposition~\ref{prop:omega-algorithm-sink}.

\begin{theorem}
\label{thm:k-algorithm}
There exists an algorithm that, given a DFA $\mathrsfs{A}=(Q,\Sigma)$, decides the winner in the $k$-game on $\mathrsfs{A}$ in $O(|Q|^4\cdot|\Sigma|)$ time.
\end{theorem}

The time bound stated in Theorem~\ref{thm:k-algorithm} is likely non-optimal.

\section{Synchronization games on fixed size automata}
\label{sec:nstates}

\subsection{Relations between $\mathbf{A}_k$ with $k\in\mathbb{N}$ and $\mathbf{A}_\omega$ for fixed size automata}

Given a DFA $\mathrsfs{A}=(Q,\Sigma)$ and $k\ge 0$, let $T_k(\mA)$ denote the set of all transformations in the transition monoid $T(\mA)$ induced by the words in $\Sigma^{<k}$. We then have the ascending chain of subsets:
\begin{equation}\label{eq:chaininT(A)}
T_0(\mA)\subseteq T_1(\mA)\subseteq\dots\subseteq T_k(\mA)\subseteq T_{k+1}(\mA)\subseteq\dotsc,
\end{equation}
whose union is the entire monoid $T(\mA)$. Since $T(\mA)$ is finite, not all inclusions in~\eqref{eq:chaininT(A)} can be strict.  Once the first equality $T_k(\mA)=T_{k+1}(\mA)$ occurs in~\eqref{eq:chaininT(A)}, the chain stabilizes permanently. (Indeed, the equality means that composing transformations from $T_k(\mA)$ with those induced by letters from $\Sigma$ creates no new transformations, so $T_{k+1}(\mA)=T_{k+2}(\mA)$, and so on.) Hence, $T(\mA)=T_k(\mA)$ for some $k\in\mathbb{N}$.

The equality $T(\mA)=T_k(\mA)$ implies that for every word $w\in\Sigma^*$, there exists a word $w'\in\Sigma^{<k}$ such that $w$ and $w'$ induce the same transformation of $Q$. Hence, if Alice has a winning strategy $\mathfrak{S}$ in the $k$-game on $\mA$, she can use $\mathfrak{S}$ to win the $\omega$-game on $\mA$ as follows: whenever Bob plays a word $w\in\Sigma^*$, she responds with a letter prescribed by $\mathfrak{S}$ in response to Bob playing $w'$. We thus conclude that if $\mA$ is an $A_k$-automaton for some sufficiently large $k\in\mathbb{N}$, then $\mA$ is also an $A_\omega$-automaton.

Here, we determine the least number $k(n)\in\mathbb{N}$ such that every $A_{k(n)}$-automaton with $n$ states is an $A_\omega$-automaton. As in the previous section, it is convenient to begin with automata with a sink and extend the analysis to the general case via Lemma~\ref{lem:2subsets}.

\begin{proposition}\label{prop:PD2}
Every $n$-state $A_{n-1}$-automaton with a sink is an $A_\omega$-automaton.
\end{proposition}

\begin{proof}
Arguing by contraposition, assume that $\mA = (Q,\Sigma)$ is a DFA with $n$ states, including a sink~$\mathbf{s}$, on which Bob has a winning strategy in the $\omega$-game. Our goal is to show that Bob can also win the $(n-1)$-game on $\mA$.

If, for each $q\in Q$, Alice had a strategy to move the token placed on $q$ to $\mathbf{s}$, then using the fact that a token on $\mathbf{s}$ never moves, she could win by removing tokens one by one from the other states. Hence, there exists a state $r\in Q$ such that, in the $\omega$-game, Bob can prevent the token initially placed on $r$ from reaching $\mathbf{s}$. Every word $w\in\Sigma^*$ can be viewed as the history of the $\omega$-game on $\mA$ up to some turn. Therefore, for every $w\in\Sigma^*$, there is a word $v\in\Sigma^*$ that Bob can play so that $(r{\cdot}w){\cdot}v\ne\mathbf{s}$. Denote $r{\cdot}w$ by $p$ and let $u=u_1\cdots u_\ell$ be a word of minimum length with $p{\cdot}u=p{\cdot}v$. Consider the sequence
\[
p,\ p{\cdot}u_1,\ p{\cdot}u_1u_2,\ \dots,\ p{\cdot}u_1\cdots u_\ell=p{\cdot}u.
\]	
This sequence omits the sink $\mathbf{s}$ since $p{\cdot}u\ne\mathbf{s}$. Moreover, it contains no repeated states. Indeed,  because if $p{\cdot}u_1\cdots u_i=p{\cdot}u_1\cdots u_j$ for some $0\le i<j\le\ell$, then the shorter word $u':=u_1\cdots u_iu_{j+1}\cdots u_\ell$ would satisfy $p{\cdot}u'=p{\cdot}v$, contradicting the minimality of~$u$. Hence, the sequence consists of $\ell+1$ distinct non-sink states. Since $\mA$ has exactly $n-1$ non-sink
states, we obtain $\ell+1\le n-1$, and therefore $|u|=\ell< n-1$.

We now describe Bob's winning strategy in the $(n-1)$-game on $\mA$. Bob follows the token $T$ initially placed on the state $r$. Whenever, after Alice's move, $T$ reaches a state $r{\cdot}w$, Bob consults his winning strategy in the $\omega$-game to choose a word $v$. He then plays a word $u$ of length less than $n-1$ such that $(r{\cdot}w){\cdot}u=(r{\cdot}w){\cdot}v$. By construction, this guarantees that the token $T$ never reaches the sink~$\mathbf{s}$, and thus at least two tokens ($T$ and the one on $\mathbf{s}$) are never removed.
\end{proof}

\begin{remark}
Let $\mathbf{D}_n^0$ denote the class of all $n$-state DFAs with a sink. Proposition~\ref{prop:PD2} implies that the hierarchy~\eqref{eq:hierarchy}, when restricted to $\mathbf{D}_n^0$, collapses at level $n-1$:
\[
\mathbf{A}_1\cap\mathbf{D}_n^0\supseteq\mathbf{A}_2\cap\mathbf{D}_n^0\supseteq\dots\supseteq\mathbf{A}_{n-1}\cap\mathbf{D}_n^0=\dots=\mathbf{A}_\omega\cap\mathbf{D}_n^0.
\]
This result is optimal in the sense that, for each $1\le k\le n-2$, the inclusion
\[
\mathbf{A}_k\cap\mathbf{D}_n^0\supseteq\mathbf{A}_{k+1}\cap\mathbf{D}_n^0
\]
is strict. This is witnessed by the DFA $\mL_m^{k}=(\mathbb{Z}_n,\,\{a_1,\dots,a_{n-1}\})$ introduced in Subsection~\ref{subsec:kspeed}. This DFA has a sink and is an $A_k$-automaton by Proposition~\ref{prop:PRTk1}. On the other hand, Bob has a simple winning strategy in the $(k+1)$-game on $\mL_m^{k}$: assuming that Alice moves first, he responds with the word $x^{k}$ whenever Alice plays a letter~$x$. Since $(n-1){\cdot}x^{k+1}=n-1$ for each $x\in\{a_1,\dots,a_{n-1}\}$, this  guarantees that the token initially placed on state $n-1$ remains on that state after each move of Bob.
\end{remark}

We now consider the general case.

 \begin{theorem}
\label{thm:k(n)}
The least number $k(n)\in\mathbb{N}$ such that every $A_{k(n)}$-automaton with $n$ states is an $A_\omega$-automaton equals $\binom{n}{2}$.
\end{theorem}

\begin{proof}
Let $\mA$ be an $A_{\binom{n}{2}}$-automaton with $n$ states. Its 2-subset automaton $\mA^{[2]}$ possesses a sink and has $\binom{n}{2}+1$ states. Besides that, by Lemma~\ref{lem:2subsets}, $\mA^{[2]}$ is an $A_{\binom{n}2}$-automaton. By Proposition \ref{prop:PD2} we conclude that $\mA^{[2]}$ is an $A_\omega$-automaton, and so is $\mA$ by Lemma~\ref{lem:2subsets}. We thus have established that $k(n)\le\binom{n}{2}$.

To prove the opposite inequality we need a series of $n$-state automata which belong to $\mathbf{A}_K\setminus\mathbf{A}_\omega$, where $K:=\binom{n}{2}-1$. We obtain such automata by adding two extra letters to the \v{C}ern\'{y} automata $\mC_n=(\mathbb{Z}_n,\{a,b\})$. Recall that the letters $a$ and $b$ act as follows:
\[
0{\cdot} a:=1,\ \ m{\cdot} a:=m\ \text{ for \ $0<m<n$,}\qquad m{\cdot} b:=m+1\kern-8pt\pmod{n}.
\]
We set $\mD_n:=(\mathbb{Z}_n,\{a,b,c,d\})$, where the newly added letters act as follows:
\[
m{\cdot}c:= \begin{cases}
			0 &\text{if } m=0,\\
			1 &\text{otherwise};
\end{cases}
\qquad
m{\cdot}d:= \begin{cases}
			1 &\text{if } m=\lceil\frac{n}2\rceil,\\
			0 &\text{otherwise}.
	\end{cases}
\]
For illustration, Figure~\ref{fig:D5} shows the automaton $\mD_5$.
\begin{figure}[hbt]
\begin{center}
\begin{tikzpicture}[scale=0.11,>=latex,auto,
  every node/.style={circle,minimum size=8mm,inner sep=0pt}]

\node[draw=blue] (n10) at (66,54) {0};
\node[draw=blue] (n11) at (44,38) {4};
\node[draw=blue] (n12) at (90,38) {1};
\node[draw=blue] (n13) at (54,18) {3};
\node[draw=blue] (n14) at (78,18) {2};

\draw[->,bend left=14] (n10) to node[sloped, above, yshift=-1ex] {$a,b$} (n12);
\draw[->,bend left=14] (n12) to node[sloped, below, yshift=1ex] {$b$} (n14);
\draw[->,bend left=14] (n14) to node[sloped, above, yshift=-1.5ex] {$c$} (n12);
\draw[->] (n14) -- node[sloped, below, yshift=1ex] {$b$} (n13);
\draw[->] (n13) -- node[sloped, below, yshift=1ex] {$b$} (n11);
\draw[->] (n11) -- node[sloped, above, yshift=-1ex] {$b,d$} (n10);
\draw[->,bend left=12] (n13) to node[sloped, above, yshift=-1ex, pos=0.4] {$c,d$} (n12);
\draw[->,bend left=14] (n12) to node[sloped, below, yshift=1ex] {$d$} (n10);
\draw[->] (n14) to node[sloped, below, yshift=1ex, pos=0.2] {$d$} (n10);
\draw[->] (n11) -- node[sloped, above, yshift=-1.5ex, pos=0.4] {$c$} (n12);

\draw[->] (n10) to [out=110,in=70,looseness=8] node[above,yshift=-1mm] {$c,d$} (n10);
\draw[->] (n11) to [out=170,in=130,looseness=8] node[above,yshift=-1mm] {$a$} (n11);
\draw[->] (n12) to [out=50,in=10,looseness=8] node[above,yshift=-1mm] {$a,c$} (n12);
\draw[->] (n13) to [out=-130,in=-170,looseness=8] node[below,yshift=1mm] {$a$} (n13);
\draw[->] (n14) to [out=-10,in=-50,looseness=8] node[below,yshift=1mm] {$a$} (n14);

\end{tikzpicture}
\caption{The DFA $\mD_5$} \label{fig:D5}
\end{center}
\end{figure}

We aim to show that the automaton $\mD_n$ belongs to $\mathbf{A}_K\setminus\mathbf{A}_{K+1}$. To this end, we need a property of the 2-subset automaton $\mC_n^{[2]}$ of the \v{C}ern\'{y} automaton $\mC_n$. This property is likely known but since we have been unable to locate any convenient reference, we include a complete proof, without claiming any originality.

First, we detail the structure of $\mC_n^{[2]}$. Under the action of $b$ on the non-sink states of $\mC_n^{[2]}$, these states, that is, the 2-element subsets $\{p,q\}\subseteq\mathbb{Z}_n$, are partitioned into \emph{orbits}, on each of which $b$ acts as a cyclic permutation. Each orbit consists of the sets $\{p,q\}$ with a fixed \emph{circular distance} defined as
\[
\min\{\ell\in\mathbb{N}: p{\cdot}b^\ell=q\ \text{ or } \ q{\cdot}b^\ell=p\}=\min\bigl\{|p-q|,\; n-|p-q|\bigr\}.
\]
Clearly, the circular distance can take any value from 1 to $\lfloor\frac{n}2\rfloor$, so there are $\lfloor\frac{n}2\rfloor$ orbits. If $n$ is even, all but one orbits have size $n$; the exceptional orbit (called \emph{antipodal}) has size $\frac{n}2$. If $n$ is odd, all orbits have size $n$. We denote by $\mathcal{O}_d$ the orbit consisting of the subsets with circular distance $d$.

Since the action of $a$ in $\mC_n$ fixes all states except 0, the action of $a$ in the 2-subset automaton $\mC_n^{[2]}$ fixes all subsets except those containing 0. For $d<\frac{n}2$, the orbit $\mathcal{O}_d$ includes two subsets containing 0: namely, $\{0,d\}$ and $\{n-d,0\}$; for $n$ even, the antipodal orbit $\mathcal{O}_{\frac{n}2}$ includes only $\{0,\frac{n}2\}$. We have
\[
\{0,d\}{\circ}a=\begin{cases}
\mathbf{s} &\text{if } d=1,
\\
\{0,d-1\}&\text{if }\ d=2,\dots,\lfloor\frac{n}2\rfloor;
\end{cases}\quad \text{and}\quad \{n-d,0\}{\circ}a=\{n-d,1\}.
\]
We see that $\{0,d\}{\circ}a\in\mathcal{O}_{d-1}$ for $d\ge 2$, while $\{n-d,0\}{\circ}a\in\mathcal{O}_{d+1}$ except for two cases: when $n$ is even and $d=\frac{n}2$ in which case $\{n-d,0\}{\circ}a=\{0,d\}{\circ}a$ goes down to $\mathcal{O}_{d-1}$ or when $n$ is odd and $d=\frac{n-1}2$ in which case $\{n-d,0\}{\circ}a$ stays in $\mathcal{O}_{d}$. Notice that each orbit $\mathcal{O}_{d}$ with $1<d<\lfloor\frac{n}2\rfloor$ is connected with both its `neighboring' orbits $\mathcal{O}_{d\pm1}$. Therefore, the union of all orbits forms a \scc{} of $\mC_n^{[2]}$.

For illustration, Figure~\ref{fig:2subsetsC6} shows the 2-subset automaton of the \v{C}ern\'{y} automaton $\mC_6$.
\begin{figure}[htb]
\begin{center}
\begin{tikzpicture}[ >=latex,shorten >=1pt,
  scale=1,
  every state/.style={circle, draw=blue, scale=1},
  every node/.style={inner sep=2pt,font=\small}]

\begin{scope}[shift={(-5,0)}]
  \node[state] (a01) at (90:1.6)  {$\{0,1\}$};
  \node[state] (a12) at (30:1.6)  {$\{1,2\}$};
  \node[state] (a23) at (-30:1.6) {$\{2,3\}$};
  \node[state] (a34) at (-90:1.6) {$\{3,4\}$};
  \node[state] (a45) at (-150:1.6){$\{4,5\}$};
  \node[state] (a50) at (150:1.6) {$\{5,0\}$};
  \node[state] (s)  [above left of=a01, xshift=-0.75cm] {$\mathbf{s}$};

  \draw[->,dashed,red] (a01) to (a12);
  \draw[->,dashed,red] (a12) to (a23);
  \draw[->,dashed,red] (a23) to (a34);
  \draw[->,dashed,red] (a34) to (a45);
  \draw[->,dashed,red] (a45) to (a50);
  \draw[->,dashed] (a50) to (a01);
  \draw[->] (a01) to (s);

\end{scope}

\begin{scope}[shift={(0,0)}]
  \node[state] (b02) at (90:1.6)  {$\{0,2\}$};
  \node[state] (b13) at (30:1.6)  {$\{1,3\}$};
  \node[state] (b24) at (-30:1.6) {$\{2,4\}$};
  \node[state] (b35) at (-90:1.6) {$\{3,5\}$};
  \node[state] (b40) at (-150:1.6){$\{4,0\}$};
  \node[state] (b51) at (150:1.6) {$\{5,1\}$};

  \draw[->,dashed,red] (b02) to (b13);
  \draw[->,dashed,red] (b13) to (b24);
  \draw[->,dashed,red] (b24) to (b35);
  \draw[->,dashed,red] (b35) to (b40);
  \draw[->,dashed] (b40) to (b51);
  \draw[->,dashed,red] (b51) to (b02);
  \draw[->,bend right=12] (b02) to (a12);
  \draw[->,bend right=25,red] (a50) to (b51);
\end{scope}

\begin{scope}[shift={(4,0)}]
  \node[state] (c03) at (90:1.3)  {$\{0,3\}$};
  \node[state] (c14) at (-30:1.3) {$\{1,4\}$};
  \node[state] (c25) at (210:1.3){$\{2,5\}$};

  \draw[->,dashed] (c03) to (c14);
  \draw[->,dashed,red] (c14) to (c25);
  \draw[->,dashed,red] (c25) to (c03);
  \draw[->,bend right=10] (c03) to (b13);
  \draw[->,bend left=23,red] (b40) to (c14);
\end{scope}

\end{tikzpicture}
\end{center}
\caption{The 2-subset automaton $\mC_6^{[2]}$. Solid and dashed edges show the action of $a$ and, resp., $b$; loops are omitted. Red-colored edges form a Hamiltonian path in the union of orbits of $\mC_6^{[2]}$.}
\label{fig:2subsetsC6}
\end{figure}
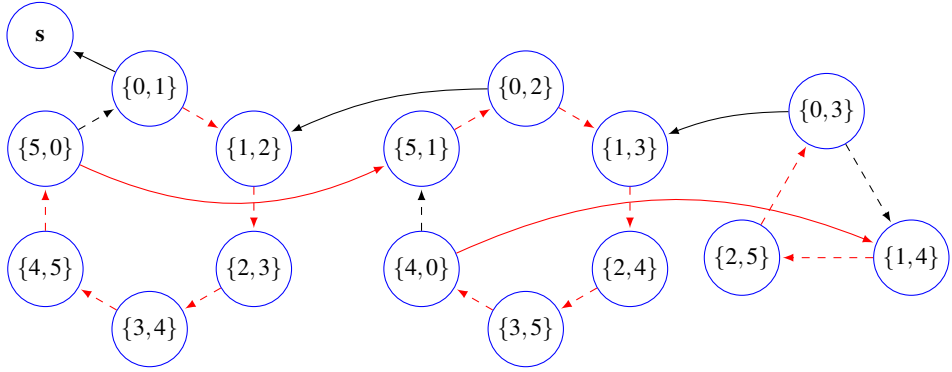

Recall that a \emph{Hamiltonian path} in a digraph is a directed path that visits each vertex of the digraph exactly once. The property of $\mC_n^{[2]}$ that we require is as follows:

\begin{lemma}\label{lem:HamiltonForward}
The union of orbits of $\mC_n^{[2]}$ has a Hamiltonian path starting at the set $\{0,1\}$, and the labels along this path form a unique word of minimum length that maps the set $\{0,1\}$ to the set $\{0,\lceil\frac{n}2\rceil\}$.
\end{lemma}

\begin{proof}
Let $u\in\{a,b\}^*$ be a word of minimum length mapping $\{0,1\}$ to $\{0,\lceil\frac{n}2\rceil\}$ (such a word exists due to the strong connectivity of the union of orbits).

The word $u$ maps a set from $\mathcal{O}_1$ to a set from $\mathcal{O}_{\lfloor\frac{n}2\rfloor}$. For each $d=1,\dots,\lfloor\frac{n}2\rfloor-1$, the only transition that moves a set in the orbit $\mathcal{O}_d$ to a set in an orbit with larger cyclic distance is the transition $\{n-d,0\}{\circ}a=\{n-d,1\}\in\mathcal{O}_{d+1}$. The shortest word mapping $\{0,1\}$ to $\{n-1,0\}$ is $b^{n-1}$, and the same word is also the shortest one mapping the `entry' set $\{n-d,1\}$ of the orbit $\mathcal{O}_d$, where $1<d<\lfloor\frac{n}2\rfloor$, to the `exit' set $\{n-d-1,0\}$ of this orbit.

It follows that the shortest word mapping $\{0,1\}$ to a set in $\mathcal{O}_{\lfloor\frac{n}2\rfloor}$ is $\left(b^{n-1}a\right)^{\lfloor\frac{n}2-1\rfloor}$ which is $\left(b^{n-1}a\right)^{\frac{n-1}2-1}$ if $n$ is odd and $\left(b^{n-1}a\right)^{\frac{n}2-1}$ if $n$ is even. The `entry' set of the orbit $\mathcal{O}_{\lfloor\frac{n}2\rfloor}$ is $(n-\lfloor\frac{n}2\rfloor+1,1)$ that can be written as $\{\lceil\frac{n}2\rceil+1,1\}$, since $n=\lfloor\frac{n}2\rfloor+\lceil\frac{n}2\rceil$. The shortest word reaching the set $\{0,\lceil\frac{n}2\rceil\}$ from $\{\lceil\frac{n}2\rceil+1,1\}$ is $b^{n-1}$ if $n$ is odd and $b^{\frac{n}2-1}$ if $n$ is even (in the latter case $\mathcal{O}_{\frac{n}2}$ is the antipodal orbit).

We conclude that the word $u$ must coincide with the word
\begin{equation}\label{eq:shortestwordinC[2]}
\overrightarrow{w}:=\begin{cases}
\left(b^{n-1}a\right)^{\frac{n-1}2-1}b^{n-1} &\text{if $n$ is odd},\\
\left(b^{n-1}a\right)^{\frac{n}2-1}b^{\frac{n}2-1} &\text{if $n$ is even}.
\end{cases}
\end{equation}

It is straightforward to verify that $|\overrightarrow{w}|=\binom{n}2-1=K$, independent of the parity of $n$.

For each $i\in\{0,1,\dots,K\}$, denote by $u_i$ the prefix of length $i$ of the word $\overrightarrow{w}$ and consider the following sequence of sets:
\begin{equation}\label{eq:setsequence}
\{0,1\}=\{0,1\}{\cdot}u_0,\  \{0,1\}{\cdot}u_1, \ \dots, \ \{0,1\}{\cdot}u_K=\{0,\left\lceil\frac{n}2\right\rceil\}.
\end{equation}
The minimality of $\overrightarrow{w}$ implies that \eqref{eq:setsequence} contains no repeated sets. The number of sets in \eqref{eq:setsequence} is $K+1=\binom{n}2$, which equals the number of 2-element subsets of $\mathbb{Z}_n$. Hence, every 2-element subset occurs in \eqref{eq:setsequence} exactly once, and the path labeled $\overrightarrow{w}$ is Hamiltonian.
\end{proof}

For illustration, Figure~\ref{fig:2subsetsC6} shows in red the Hamiltonian path in the union of orbits of $\mC_6^{[2]}$ constructed according to the recipe of Lemma~\ref{lem:HamiltonForward}.

\medskip

Now we are ready to prove that the automaton $\mD_n$ belongs to $\mathbf{A}_K\setminus\mathbf{A}_{K+1}$.

To win the $K$-game on $\mD_n$, Alice plays the letter $c$ on her first turn, thereby compressing the entire state set to $\{0,1\}$. If Bob responds with a word $v$ of length less than $K$, Lemma~\ref{lem:HamiltonForward} guarantees that $\{0,1\}{\cdot}v\ne\{0,\lceil\frac{n}2\rceil\}$. If $\{0,1\}{\cdot}v$ omits 0, Alice wins by playing~$c$; if $\{0,1\}{\cdot}v$ omits $\lceil\frac{n}2\rceil$, she wins by playing $d$.

Bob's winning strategy in the $(K+1)$-game on $\mD_n$ is to maintain tokens on states 0 and $\lceil\frac{n}2\rceil$, which he can always achieve. If Alice moves the tokens from the set $\{0,\lceil\frac{n}2\rceil\}$ to some set $\{p,q\}\subseteq\mathbb{Z}_n$, Bob locates the latter set in the sequence \eqref{eq:setsequence}, where it appears as $\{0,1\}{\cdot}u_i$ for some $i\in\{0,1,\dots,K\}$. Then
$\{p,q\}{\cdot}v_i=\{0,\lceil\frac{n}2\rceil\}$, where $v_i$ is the suffix of $w$ following $u_i$. Since $|v_i|\le|w|=K$, Bob can play $v_i$, thereby returning the tokens to states 0 and $\lceil\frac{n}2\rceil$.
\end{proof}

\begin{remark}
In terms of the hierarchy~\eqref{eq:hierarchy} restricted to the class $\mathbf{D}_n$ all $n$-state DFAs, Theorem~\ref{thm:k(n)} means that
\[
\mathbf{A}_1\cap\mathbf{D}_n\supseteq\mathbf{A}_2\cap\mathbf{D}_n\supseteq\dots\supseteq\mathbf{A}_{\binom{n}{2}-1}\cap\mathbf{D}_n\supsetneqq\mathbf{A}_{\binom{n}{2}}\cap\mathbf{D}_n=\dots=\mathbf{A}_\omega\cap\mathbf{D}_n.
\]
In fact, for $1\le k\le\binom{n}{2}-2$, all inclusions
\[
\mathbf{A}_k\cap\mathbf{D}_n\supseteq\mathbf{A}_{k+1}\cap\mathbf{D}_n
\]
are strict. This is witnessed by the automata $\mL_m^{k}$ for $1\le k\le n-2$ (see the remark following Proposition~\ref{prop:PD2}) and by the automaton $\mE_n$ for $k=n-1$ (see Lemma~\ref{lem:strict}). For each $n\le k\le \binom{n}{2}-2$, a witness is obtained by modifying the automaton from $\mathbf{A}_K\setminus\mathbf{A}_{K+1}$ constructed in the proof of Theorem~\ref{thm:k(n)}. In the following description of this modification, we use the notation from that proof; in particular, $K=\binom{n}{2}-1$.

Fix $k\in\{n,n+1,\dots,K\}$ and let $\{p_k,q_k\}:=\{0,1\}{\cdot}u_k$, where $u_k$ is the length $k$ prefix of the word $\overrightarrow{w}$ defined in \eqref{eq:shortestwordinC[2]}. (For instance, $\{p_n,q_n\}=\{0,1\}{\cdot}u_n= \{n-1,1\}$ and $\{p_K,q_K\}=\{0,1\}{\cdot}u_K=\{0,\left\lceil\frac{n}2\right\rceil\}$.) Consider the DFA $\mD_n^k:=(\mathbb{Z}_n,\{a,b,c_k,d_k\})$ obtained by adding to the \v{C}ern\'{y} automaton $\mC_n=(\mathbb{Z}_n,\{a,b\})$ two extra letters $c_k$ and $d_k$ that act as follows:
\[
m{\cdot}c_k:= \begin{cases}
			0 &\text{if } m=p_k,\\
			1 &\text{otherwise};
\end{cases}
\qquad
m{\cdot}d_k:= \begin{cases}
			1 &\text{if } m=q_k,\\
			0 &\text{otherwise}.
	\end{cases}
\]
In particular, $\mD_n^K$ coincides with the automaton $\mD_n\in\mathbf{A}_K\setminus\mathbf{A}_{K+1}$. The proof that the DFA $\mD_n^k$ lies in $\mathbf{A}_k\setminus\mathbf{A}_{k+1}$ for each $n\le k\le \binom{n}{2}-2$ essentially repeats the proof of Theorem~\ref{thm:k(n)} modulo Lemma~\ref{lem:HamiltonForward}, and is therefore omitted.
\end{remark}

\subsection{$m/\omega$-games on fixed size \sa}
Recall that in an $m/\omega$-game, Alice may play nonempty words of length at most $m$, whereas Bob may play arbitrary words. Here, we show that the least number $m(n)\in\mathbb{N}$ such that Alice can win the $m(n)/\omega$-game on every \san{} with $n$ states is $\binom{n}2$. Thus, $m(n)$ coincides with the parameter $k(n)$ considered in the preceding subsection. This may seem somewhat surprising, since the parameters $m(n)$ and $k(n)$ appear to be completely unrelated at first glance. In fact, they are related: both are equal to the diameter of the 2-subset automaton of the \v{C}ern\'{y} automaton $\mC_n$.

\begin{proposition}
\label{prop:m(n)}
The least number $m(n)\in\mathbb{N}$ such that Alice can win the $m(n)/\omega$-game on every \san{} with $n$ states equals $\binom{n}{2}$.
\end{proposition}

\begin{proof}
The standard argument showing that $m(n)\le\binom{n}{2}$ is included for completeness only.

Let $\mA=(Q,\Sigma)$ be a \san{} with $n$ states. As in Lemma~\ref{lem:localization}, it suffices to prove that Alice can win the $\binom{n}{2}/\omega$-game on $\mA$ from every position with two tokens. Consider such a position $P$ in which states $q$ and $q'$ hold tokens. Since the DFA $\mA$ is synchronizing, there exists a word whose action brings the tokens to same state. Let $v\in\Sigma^*$ be such a word of minimum length and for $i=0,1\dots,|v|$, let $v_i$ denote the prefix of $v$ of length $i$. The minimality of $v$ implies that all sets in the sequence
\[
\{q,q'\}=\{q,q'\}{\cdot}v_0,\ \{q,q'\}{\cdot}v_1, \dots, \{q,q'\}{\cdot}v_{|v|-1},
\]
are distinct and non-singletons. Hence, the length $|v|$ of this sequence does not exceed the number of 2-element subsets of $Q$, namely, $\binom{n}{2}$. Therefore, $v$ is a legitimate move for Alice in $\binom{n}{2}/\omega$-game on $\mA$, and by playing this word Alice instantly wins the game from the position $P$.

On the other hand, Bob can win the $K/\omega$-game on the \v{C}ern\'{y} automaton~$\mC_n$, where $K=\binom{n}{2}-1$. To show this, we use a Hamiltonian path in $\mC_n^{[2]}$ that is, in a sense, opposite to the Hamiltonian path of Lemma~\ref{lem:HamiltonForward}.

\begin{lemma}\label{lem:HamiltonBackward}
The automaton $\mC_n^{[2]}$ has a Hamiltonian path starting at the set $\{\lceil\frac{n}2\rceil+1,1\}$, and the labels along this path form a unique word of minimum length that maps the set $\{\lceil\frac{n}2\rceil+1,1\}$ to the sink of $\mC_n^{[2]}$.
\end{lemma}

For illustration, Figure~\ref{fig:2subsetsC6back} shows in red the Hamiltonian path of Lemma~\ref{lem:HamiltonBackward} in $\mC_6^{[2]}$.

\begin{figure}[htb]
\begin{center}
\begin{tikzpicture}[ >=latex,shorten >=1pt,
  scale=1,
  every state/.style={circle, draw=blue, scale=1},
  every node/.style={inner sep=2pt,font=\small}]

\begin{scope}[shift={(-5,0)}]
  \node[state] (a01) at (90:1.6)  {$\{0,1\}$};
  \node[state] (a12) at (30:1.6)  {$\{1,2\}$};
  \node[state] (a23) at (-30:1.6) {$\{2,3\}$};
  \node[state] (a34) at (-90:1.6) {$\{3,4\}$};
  \node[state] (a45) at (-150:1.6){$\{4,5\}$};
  \node[state] (a50) at (150:1.6) {$\{5,0\}$};
  \node[state] (s)  [above left of=a01, xshift=-0.75cm] {$\mathbf{s}$};

  \draw[->,dashed] (a01) to (a12);
  \draw[->,dashed,red] (a12) to (a23);
  \draw[->,dashed,red] (a23) to (a34);
  \draw[->,dashed,red] (a34) to (a45);
  \draw[->,dashed,red] (a45) to (a50);
  \draw[->,dashed,red] (a50) to (a01);
  \draw[->,red] (a01) to (s);

\end{scope}

\begin{scope}[shift={(0,0)}]
  \node[state] (b02) at (90:1.6)  {$\{0,2\}$};
  \node[state] (b13) at (30:1.6)  {$\{1,3\}$};
  \node[state] (b24) at (-30:1.6) {$\{2,4\}$};
  \node[state] (b35) at (-90:1.6) {$\{3,5\}$};
  \node[state] (b40) at (-150:1.6){$\{4,0\}$};
  \node[state] (b51) at (150:1.6) {$\{5,1\}$};

  \draw[->,dashed] (b02) to (b13);
  \draw[->,dashed,red] (b13) to (b24);
  \draw[->,dashed,red] (b24) to (b35);
  \draw[->,dashed,red] (b35) to (b40);
  \draw[->,dashed,red] (b40) to (b51);
  \draw[->,dashed,red] (b51) to (b02);
  \draw[->,bend right=12,red] (b02) to (a12);
  \draw[->,bend right=25] (a50) to (b51);
\end{scope}

\begin{scope}[shift={(4,0)}]
  \node[state] (c03) at (90:1.3)  {$\{0,3\}$};
  \node[state] (c14) at (-30:1.3) {$\{1,4\}$};
  \node[state] (c25) at (210:1.3){$\{2,5\}$};

  \draw[->,dashed] (c03) to (c14);
  \draw[->,dashed,red] (c14) to (c25);
  \draw[->,dashed,red] (c25) to (c03);
  \draw[->,bend right=10,red] (c03) to (b13);
  \draw[->,bend left=23] (b40) to (c14);
\end{scope}

\end{tikzpicture}
\end{center}
\caption{The 2-subset automaton $\mC_6^{[2]}$. Solid and dashed edges show the action of $a$ and, resp., $b$; loops are omitted. Red-colored edges form a Hamiltonian path in $\mC_6^{[2]}$.}
\label{fig:2subsetsC6back}
\end{figure}
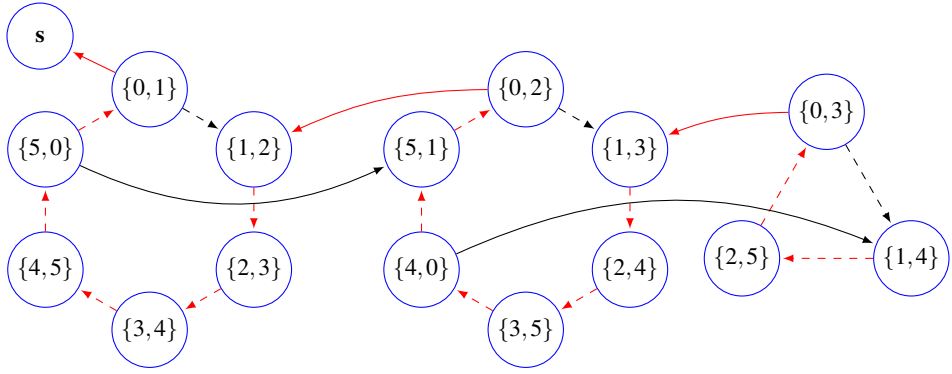

The proof of Lemma~\ref{lem:HamiltonBackward} is very similar to that of Lemma~\ref{lem:HamiltonForward} and is therefore omitted. We only mention that the word whose existence is claimed in the lemma is
\[
\overleftarrow{w}:=\begin{cases}
\left(b^{n-1}a\right)^{\frac{n-1}2} &\text{if $n$ is odd},\\
b^{\frac{n}2-1}a\left(b^{n-1}a\right)^{\frac{n}2-1} &\text{if $n$ is even}.
\end{cases}
\]
One can observe that the word $\overleftarrow{w}$ is obtained by appending $a$ to the word $\overrightarrow{w}$ of \eqref{eq:shortestwordinC[2]}, read backwards.

Bob's winning strategy in the $K/\omega$-game on $\mC_n$ is to maintain tokens on states 1 and $\lceil\frac{n}2\rceil+1$, which he can always achieve. Indeed, Lemma~\ref{lem:HamiltonBackward} implies that $\{\lceil\frac{n}2\rceil+1,1\}{\circ}u\ne\mathbf{s}$ for every word $u$ of length at most $K$. Hence, after Alice's move, the tokens originating from the set $\{\lceil\frac{n}2\rceil+1,1\}$ arrive at some set $\{p,q\}\subseteq\mathbb{Z}_n$, and Bob can send them back to states 1 and $\lceil\frac{n}2\rceil+1$ due to the strong connectivity of the union of orbits.
\end{proof}

\begin{remark}
For the case of automata with a sink, one can easily prove that Alice can win the $(n-1)/\omega$-game on every $n$-state \san{} with a sink. The bound $n-1$ is tight which is witnessed by the DFA $\mL_{n-2}^{1}$ from the series $\mL_m^{k}$ introduced in Subsection~\ref{subsec:kspeed}.
\end{remark}

\section{Conclusion and future work}
\label{sec:conclude}

We summarize our results and highlight questions that have so far remained open.

Theorem~\ref{thm:main}, together with the remark following it, provides the exact value of the \v{C}ern\'{y} function for the class of all $A_\omega$-automata. For $A_k$-automata with $k\in\mathbb{N}$, we currently have only a lower bound, obtained by combining Propositions~\ref{prop:PRTk1} and~\ref{prop:quadratic}. For each $k\in\mathbb{N}$, the class $\mathbf{A}_k$ is not contained in any class of \sa{}, for which a quadratic upper bound on the \v{C}ern\'{y} function has been established so far; see \cite{List} for a catalogue of such classes. Therefore, finding an $O(n^2)$ upper bound on $\mathfrak{C}_{\mathbf{A}_k}(n)$ is of significant interest.

We believe that the lower bound $\frac{n(n-1)}{2k}$ for $\mathfrak{C}_{\mathbf{A}_k}(n)$, with $k\in\mathbb{N}$, can be improved. Recall that the DFAs witnessing this bound are automata with a sink. For several classes of automata with known lower bounds on their \v{C}ern\'{y} functions, the bound for the subclass consisting of automata with a sink is around half of the bound for the entire class. If this empirical observation extends to the class $\mathbf{A}_k$, one might therefore expect a lower bound with leading term $\frac{n^2}k$ rather than $\frac{n^2}{2k}$.

With regard to \rl{}s of $A_k$-automata, we also pose two particular questions.
\begin{question}\label{que:nnn}
Is the reset threshold of every $A_k$-automaton with $k$ states less than $k$?
\end{question}

The answer is affirmative for automata with a sink, by Proposition~\ref{prop:PD2} and Theorem~\ref{thm:main}.

\begin{question}\label{que:CD}
Do there exist real positive constants $C$ and $D$ such that $\rt(\mA)\le Cn$ whenever $\mA$ is an $A_k$-automaton with $n$ states and $k\ge Dn$?
\end{question}

Proposition~\ref{prop:iteration} gives the upper bound $m(n-2)+1$ on the \rl{}s of $m/\omega$-au\-tom\-ata with $n$ states. The bound is tight for $m=2$ but is likely not tight for larger $m$. For the case $m=3$, there is some evidence that the bound can be improved to $3n-6$, and we therefore pose the following question.
\begin{question}\label{que:m=3}
Does every  $3/\omega$-automaton with $n$ states have a \sw\ of length $3n-6$?
\end{question}

Of interest is the possible length of a $k$-game on an $A_k$-automaton with $n$ states. For $k=\omega$, Theorem~\ref{thm:main} provides the upper bound $n-1$, which is tight since the length of any $k$-game on a DFA $\mA$ is at least $\rt(\mA)$. However, for $k\in\mathbb{N}$, we currently have only the straightforward cubic upper bound given by Corollary~\ref{cor:cubic}. It is hard to imagine a series of $A_k$-automata with $n$ states on which Alice needs $\Theta(n^3)$ moves to win, yet the existence of such automata would contradict neither current knowledge nor the \v{C}ern\'{y} conjecture. We believe that the following question may be quite challenging.
\begin{question}\label{que:length}
Do there exist a $1<k<\omega$ and an infinite series of $A_k$-automata whose $k$-game lengths grow faster than any quadratic function of the state number?
\end{question}

Notice that, for $k\in\mathbb{N}$, there is no immediate relation between the \rl{} of an $A_k$-automaton and the number of moves Alice needs to win on it, except that the latter is always greater than or equal to the former. Using the duplication construction from \cite[Theorem 6]{FMV}, one can easily produce infinite series of $2n$-state $A_k$-automata with $1<k<\omega$ and \rl{} equal to 2, on which Alice needs $\Theta(n^2)$ moves to win.

As for algorithmic issues, we believe that the upper bound on the time complexity of recognizing $A_k$-automata given by Theorems~\ref{thm:omega-algorithm} and~\ref{thm:k-algorithm} is tight for $k=\omega$ and may be improved for $k<\omega$. Observe that the bound of Theorem~\ref{thm:k-algorithm} is independent of the parameter $k$, and may therefore be unnecessarily pessimistic when $k$ is a constant, which is a natural setting, unless, as in Section~\ref{sec:nstates}, one compares $k$-games for various values of $k$. This suggests approaching the problem of recognizing $A_k$-automata with $k<\omega$ under the framework of parameterized complexity.

\end{document}